\documentclass[journal]{IEEEtran}
\IEEEoverridecommandlockouts

\usepackage{bm,mathtools,amsmath,amsfonts,amssymb,amsthm,bbm,float,hyperref,euscript,setspace,tikz,steinmetz,subcaption,cite,enumitem}

\usepackage[normalem]{ulem}
\usepackage[mode=buildnew]{standalone}
\usepackage{mathabx} 
\newtheorem{theorem}{Theorem}
\newtheorem{lemma}{Lemma}

\newtheorem{definition}{Definition}
\newtheorem{remark}{Remark}

\usepackage{mleftright}\mleftright

\definecolor{darkgreen}{rgb}{0, 0.5, 0}
\definecolor{darkred}{RGB}{128, 0, 0}

\newcommand{\snr}{\textnormal{SNR}} 

\newcommand{\db}{\textnormal{dB}}
\newcommand{\stkout}[1]{\ifmmode\text{\sout{\ensuremath{#1}}}\else\sout{#1}\fi}

\newcommand{\vc}[1]{\mathbf{#1}} 
\newcommand{\der}{\mathrm{d}} %
\newcommand{\widesim}[2][1.5]{\mathrel{\overset{#2}{\scalebox{#1}[1]{$\sim$}}}}
\newcommand{\norm}[1]{\|#1\|}
\newcommand{\abs}[1]{|#1|}

\newcommand{\ra}{\rightarrow}
\newcommand{\eg}{\emph{e.g.,}}
\newcommand{\ie}{\emph{i.e.,}}
\DeclareMathOperator{\sinc}{sinc}
\DeclareMathOperator{\diag}{diag}

\newcommand{\iid}{i.i.d.}

\newcommand{\eqdef}{\stackrel{\Delta}{=}}

\title{Lower Bound on the Capacity of the Continuous-Space SSFM Model of Optical Fiber}

\author{Milad Sefidgaran and Mansoor Yousefi\thanks{M. Sefidgaran and M. Yousefi are with the Communications and Electronics Department of Télécom Paris (Institut Polytechnique de Paris), Paris, France. E-mails:\{milad.sefidgaran,yousefi\}@telecom-paris.fr.\\This paper was presented in part at ITW 2020.}}

\begin{document}
\maketitle

\begin{abstract}  
The capacity of a discrete-time model of optical fiber described by the split-step Fourier method (SSFM) as a function of  the signal-to-noise ratio $\snr$ and the number of segments in distance $K$ is considered.
It is shown that if $K\geq \snr^{2/3}$ and $\snr \ra\infty$, the capacity of the resulting continuous-space lossless model is lower bounded by 
$\frac{1}{2}\log_2(1+\snr) - \frac{1}{2}+ o(1)$,
where $o(1)$ tends to zero with $\snr$. 
As $K\ra\infty$, the inter-symbol interference (ISI) averages out to zero due to the law of large numbers 
and the SSFM model tends to a diagonal phase noise model.
It follows that, in contrast to the discrete-space model where there is only one signal degree-of-freedom (DoF) at high powers,  
the number of DoFs in the continuous-space model is at least half of the input dimension $n$.
Intensity-modulation and direct detection achieves this rate. 
The pre-log in the lower bound when $K= \sqrt[\delta]{\snr}$ is generally characterized in terms of $\delta$. 

It is shown that if the nonlinearity parameter $\gamma\ra\infty$, the capacity  of the continuous-space model is $\frac{1}{2}\log_2(1+\snr)+ o(1)$.

The SSFM model when the dispersion matrix does not depend on $K$ is considered. 
It is shown that the  capacity of this model when $K= \sqrt[\delta]{\snr}$, $\delta>3$, and $\snr \ra \infty$ is  $\frac{1}{2n}\log_2(1+\snr)+ O(1)$. Thus, there is only one DoF in this model. 

Finally, it is found that the maximum achievable information rates (AIRs) of the SSFM model with back-propagation equalization obtained using numerical simulation follows a double-ascent curve.
The AIR characteristically increases with $\snr$, reaching a peak at a certain optimal power, and then decreases as $\snr$ is further increased. The peak is attributed to a balance between noise and stochastic ISI.  However,  if the power is further increased, the AIR will increase again, approaching the lower bound $\frac{1}{2}\log(1+\snr)- \frac{1}{2} + o(1)$. The second ascent is because the ISI averages out to zero with $K\ra\infty$ sufficiently fast.

\end{abstract}

\begin{IEEEkeywords}
Optical fiber, channel capacity, split-step Fourier method.
\end{IEEEkeywords}


\section{Introduction}
Optical fiber is the medium of choice for high-speed data transmission.  Although general expressions for the capacity of discrete-time point-to-point channels are derived in \cite{Shan49,VerduHan94}, evaluating these expressions for models of optical fiber remains difficult.

Optical fiber is modeled by the stochastic nonlinear Schr\"odinger (NLS) equation. 
There are two effects in the channel that impact the capacity. 
First, nonlinearity transforms additive noise to \emph{phase noise} during the propagation. 
As the amplitude of the input signal tends to infinity, 
the phase of the output signal tends to a uniform random variable in
the zero-dispersion channel \cite[Sec.~IV]{yousefi2011opc}. 
Second, dispersion converts phase noise to amplitude noise introducing a \emph{multiplicative noise}.
The successive application of the phase and multiplicative noise makes signal noise interaction intractable.

The achievable information rates (AIRs) of the wavelength-division multiplexing (WDM) vanish  at high powers due to treating interference (arising from the application of the linear multiplexing to the nonlinear channel) as noise
\cite{SplettKurzke93, mitra2000nli,essiambre2010clo, SecondinitFores13}. On the other hand, it is shown that the capacity $\mathcal C(\snr,K)$ of the discrete-time models of optical fiber 
as a function of the signal-to-noise ratio (SNR) and the number of segments in distance $K$ satisfies \cite{yousefi2015cwit,kramer2015upper}
\begin{IEEEeqnarray}{rCl}
\mathcal C(\snr,K) \leq \log_2(1+\text{SNR}).
\label{eq:ub}
\end{IEEEeqnarray}

The problem of finding the capacity has been investigated for the non-dispersive case in \cite{turitsyn2003ico,yousefi2011opc,ReznichenkoTerekhov20,Keykhosravi19}. It is shown that the asymptotic capacity of this channel is $\frac{1}{2}\log_2(\mathcal{P})+o(1)$ \cite{yousefi2011opc}, where $\mathcal P$ is the average input signal power.

The stochastic NLS equation can be discretized using the split-step Fourier method (SSFM). 
The capacity of the discrete-time discrete-space SSFM model of the optical fiber with fixed step size in distance as a function of $\snr$ is studied in \cite{yousefi2016cap}. 
It is shown that this model tends to a linear fading channel as $\snr \ra \infty$,  described by a random matrix $\mathsf{M}_K$. The asymptotic capacity of this model is \cite[Thm.~1]{yousefi2016cap}
\begin{IEEEeqnarray}{c}
 \mathcal{C}(\snr, K){=} 
 \begin{cases}
 \IEEEstrut
  \frac{1}{2n}\log_2\left(\snr \right) + O(1),    & \textnormal{const. loss},\\[1pt]  \frac{1}{n}\log_2\log_2\left(\snr\right) + O(1),& \textnormal{non-const. loss},
 \IEEEstrut
 \end{cases}
 \IEEEeqnarraynumspace
 \label{eq:asym-cap}
\end{IEEEeqnarray}
where $n$ is the dimension of the input vector, and the loss coefficient is considered as a function of frequency. As a result, there is only one signal degrees-of-freedom (DoF) at high powers (signal energy) due to signal-noise interaction. However, the model in \cite{yousefi2016cap} may not describe realistic fiber where 
the distance is continuous.

The capacity of the discrete-time discrete-space SSFM model as a function of $\snr$ and the number of segments in distance $K$ is studied in \cite{kamran2015bound}. It seems that the analysis in \cite{kamran2015bound} suggests that if $K$ tends to infinity sufficiently fast as $K = \sqrt[4]{\snr}$ and $\snr \ra\infty$, 
the capacity is lower bounded by $\frac{1}{8}\log_2(1+\snr)+c$ where $c<\infty$.

In this paper, we consider the SSFM model of optical fiber as a function of $K$ and $\snr$. The contributions of the paper are as follows.

\begin{itemize}[leftmargin=*,wide]
\item[\textit{a})] 
First, we show that when $K\geq \snr ^{2/3}$ and $\snr \ra\infty$, 
the off-diagonal terms in the random matrix $\mathsf{M}_K$ in \cite{yousefi2016cap}, representing the stochastic inter-symbol interference (ISI), tend to zero due to the law of large numbers, and $\mathsf{M}_K$ tends to a diagonal matrix with phase noise.
As a consequence,  the capacity of the lossless continuous-space SSFM model is lower bounded as \begin{IEEEeqnarray}{rCl}
     \mathcal{C}(\snr)&\triangleq &\lim \limits_{K\ra \infty} \mathcal{C}(\snr,K)
      \nonumber\\
     &\geq& \frac{1}{2}\log_2\left(1+ \snr\right)- \frac{1}{2}+o(1), 
     \label{eq:mainThSSFM}
\end{IEEEeqnarray}
where the term $o(1)$ tends to zero with $\snr\ra\infty$. This suggests that, unlike the discrete-space SSFM model where asymptotically there is only one DoF and the capacity is essentially finite (for large $n$), in the continuous-space model the number of DoFs is at least half of the input dimension. In particular, the capacity grows with the input power with pre-log of at least $1/2$. The pre-log in the lower bound when $K=\sqrt[\delta]{\snr}$ is generally characterized in terms of $\delta$.

\item[\textit{b)}] Second, we consider the SSFM model when the nonlinearity parameter $\gamma \ra \infty$. 
It is shown that this channel is a fading channel for any $K$ and $\snr$. As a result, when $K\ra\infty$, the channel simplifies to $n$ independent phase noise channels in the lossless case, with the capacity $\mathcal{C}(\snr)=\frac{1}{2}\log_2(1+\snr)-\frac{1}{2}+o(1)$.

\item[\textit{c)}] Third, we consider the lossless SSFM model in which the dispersion matrix does not depend on $K$. It is shown that when $K=\sqrt[\delta]{\snr}$, $\delta>3$, and $ \snr \ra \infty$, the capacity is  $\frac{1}{2n}\log_2(1+\snr)+O(1)$.  
In this case, there is one DoF asymptotically as in \eqref{eq:asym-cap} \cite{yousefi2016cap}.

\item[\textit{d)}] Finally, we simulate the AIR of the SSFM model with back-propagation equalization. As previously observed, the AIR characteristically increases with $\snr$, reaching a peak at a certain optimal power, and then decreases as $\snr$ is further increased (typically to near zero in WDM). The peak is attributed to a balance between noise and ISI.  However,  if the power is increased further, the AIR will increase again, approaching the $\frac{1}{2}\log(1+\snr)- \frac{1}{2} + o(1)$ lower bound. The second ascent is because the ISI vanishes as 
$K\ra\infty$ sufficiently fast.

\end{itemize}

The paper is organized as follows. 
The notation is introduced in Section~\ref{sec:Notations}. The 
discrete- and continuous-space SSFM models are presented in Section~\ref{sec:model}. 
The main capacity lower bound is presented in Section~\ref{sec:CapacityResults}, which is proved and extended in Sections~\ref{sec:CapacityResultsExt} and \ref{sec:proof}. The results are verified by numerical simulations in Section~\ref{sec:simulations}, and the paper is concluded in Section \ref{sec:conc}.   Appendix~\ref{sec:math}. provides background on a few mathematical concepts.


\section{Notation} \label{sec:Notations}
Real and complex numbers are denoted by $\mathbb{R}$ and $\mathbb{C}$, respectively, with the imaginary unit $j=\sqrt{-1}$. The real and imaginary parts of a complex number $x$ are denoted by $\mathfrak{R}(x)$ and $\mathfrak{I}(x)$, respectively. The magnitude and phase of $x\in\mathbb C$ are denoted by $\abs{x}$ and $\angle{x}$.
The complex conjugate of $x\in\mathbb C$ is $x^*$.
Important scalars are shown with the calligraphic font, \eg\ $\mathcal{P}$ for power, $\mathcal{C}$ for the capacity, and $\mathcal{L}$ for the length of  optical fiber. 
 

Bold letters are used to denote vectors, \eg\ $\vc{x}$. The $p$-norm of a vector $\vc{x} \in \mathbb{C}^n$ is 
\begin{IEEEeqnarray}{c}
    \norm{\vc{x}}_p=\bigl( \abs{x_1}^p +    \abs{x_2}^p + \cdots + \abs{x_n}^p \bigr)^{1/p}.
    \end{IEEEeqnarray}
The Euclidean norm with $p=2$ is $\norm{\vc{x}} \eqdef\norm{\vc{x}}_2$.

The entries of a  sequence of vectors $\vc{x}_i\in\mathbb C^n$, $i=1,2,\ldots$, are indexed with convention
\begin{IEEEeqnarray}{C}
\vc x_{i} = \begin{pmatrix}
x_{i,1}, x_{i,2}, \ldots, x_{i,n}
\end{pmatrix}.
\label{eq:indexing}
\end{IEEEeqnarray}

The $n$-sphere is denoted by $\mathcal{S}^n$. A vector $\vc{x}$ in the spherical coordinate is represented by its norm $\norm{\vc{x}}$ and its direction $\hat{\vc{x}}=\vc{x}/\norm{\vc{x}}$. The spherical coordinate system is introduced in Appendix~\ref{sec:math}. 

Random variables and their realizations are represented by the upper- and lower-case letters respectively. The probability density function (PDF) of a random variable $X$ is denoted by $P_X(x)$. The expected value of a random variable $X$ is denoted by $\mathbb{E}[X]$. The uniform distribution on the interval $[a,b)$ is denoted by $\mathcal{U}(a,b)$. The PDF of a zero-mean circularly-symmetric complex Gaussian random vector with covariance matrix $\mathrm{K}$ is denoted by $\mathcal{N}_{\mathbb{C}}(0,\mathrm{K})$. Equality of random variables $X$ and $Y$ in distribution is written as 
$X\stackrel{d}{=}Y$.

A random variable in $\mathbb C^n$ is said to be absolutely continuous if its PDF is bounded and has at least one finite moment  \cite[Def.~3]{GhourchianGohariAmini17}. Such random variable has an absolutely continuous density with respect to the Lebesgue measure, and its PDF does not include a Dirac delta function.

Let $\mathcal X\subseteq\mathbb R^n$ and $g:\mathcal X\mapsto \mathbb R$.
A sequence of probability distributions  $(\mu_{\mathcal P})_{\mathcal P\in\mathbb{R}^+}$ on $\mathcal X$ with the average cost constraint $\mathbb E g(\vc X) \leq \mathcal P$ is said to escape to infinity
with $\mathcal P $  if \cite[Def.~2.4]{moser2004dbb}
\begin{IEEEeqnarray*}{c}
\lim \limits_{\mathcal{P} \ra \infty} \mu_{\mathcal{P}} \left( \vc{x}\in\mathcal X\colon g(\vc{x}) \leq\mathcal P_0\right)=0,
\end{IEEEeqnarray*}
for any $\mathcal P_0>0$. 

We say a sequence of channels with conditional distributions $(P_a(\vc y|\vc x))_{a\in\mathbb R}$, $\vc x, \vc y\in \mathcal X$, tends to a channel $Q(\vc y|\vc x)$ as $a\ra\infty$ if $\lim_{a\ra\infty}P_a(\vc y| \vc x) = Q(\vc y| \vc x)$ point-wise for all $\vc x$ and $\vc y$. We say a channel $P(\vc y| \vc x)$ tends to a channel $Q(\vc y|\vc x)$ as the  distribution of $\vc X$ escapes to infinity with an average cost $\mathcal P$, if the output of $P$ tends to the output of $Q$ in probability as $\mathcal P\ra\infty$ for any sequence of input distributions that escapes to infinity. When the cost function is $g(\vc x)=\norm{\vc x}_2^2$, roughly speaking this implies that $P( \vc y| \vc x)\ra Q(\vc y| \vc x)$ point-wise, for all $\vc y$ and all $\vc x$ with $\norm{\vc x}_2>c$ for all $c>0$ (except possibly on an input set with zero measure).

A  sequence  of $n$ numbers $x_1,\ldots, x_{n}$ is shown as $(x_{\ell})_{\ell=1}^n$ or $x^n$. The set of integers $\{1,2,\ldots,n\}$ is denoted by $[n]$. A sequence of independent and identically distributed (i.i.d.) random variables drawn from the PDF $P_X(x)$ is presented as $X_{\ell} \widesim{\text{i.i.d.}} P_X(x)$, $\ell=1,2,3,\ldots$.
 
Deterministic matrices are denoted by upper-case letters with mathrm font, \eg\ $\mathrm{D}$, and random matrices are shown by upper-case letters with mathsf font, \eg\ $\mathsf{M}$. The identity matrix of size $n$ is $\mathrm{I}_n$. 

For a sequence of matrices $(A_i)_{i=1}^m$, the product is defined with convention $\prod \limits_{i=1}^m \mathrm{A}_i=\mathrm{A}_m  \mathrm{A}_{m-1} \cdots  \mathrm{A}_1$.  A diagonal matrix $\mathsf{R}$ with
diagonal entries $R_i$ is denoted by $\mathsf R=\diag( R_1,\ldots, R_n)$. The following diagonal matrix is used throughout the paper
\begin{IEEEeqnarray}{c}
\mathsf{R}(\boldsymbol{\theta})\triangleq
\diag \Bigl( \bigl(\exp(j\theta_{\ell})\bigr)_{{\ell}=1}^n \Bigr), \label{def:diagMat}
\end{IEEEeqnarray}
where $\theta_{\ell}\widesim{\text{i.i.d.}} \mathcal{U}(0,2\pi)$.

The group of complex-valued $n\times n$ unitary matrices  is denoted by $\mathbb{U}_n$. Some properties  of $\mathbb{U}_n$  are reviewed in Appendix~\ref{sec:math}.

Suppose that $\nu$ is an equivalence relation on the set $\mathcal{I}_n=\{1,2,\ldots,n\}$, partitioning it into non-empty equivalence classes $\mathcal{I}_1,\mathcal{I}_2,\ldots,\mathcal{I}_{m(\nu)}$.
The notation $\mathbb{U}_n(\nu)$ is used to denote the group of block diagonal unitary matrices $\mathrm{A}$ in which if the integers $r$ and $s$ do not belong to one class, then $\mathrm{A}_{r,s}=0$. 

Given two functions $f(x)\colon \mathbb{R}\rightarrow \mathbb{C}$ and $g(x)\colon\mathbb{R}\rightarrow \mathbb{C}$, we say $f(x)=O\left(g(x)\right)$, if there exists a finite  $c>0$ and $x_0>0$ such that $|f(x)|\leq c\,|g(x)|$, for all $x \geq x_0$.  In addition, $f(x)=o\left(g(x)\right)$ if for any  $c>0$ there exists a finite  $x_0>0$ such that  $|f(x)|\leq c \,|g(x)|$, for all $x \geq x_0$. 

For a sequence of scalar random variables $\left(X_k\right)_{k=1}^{\infty}$ and constants $\left(a_k\right)_{k=1}^{\infty}$, where $X_k,a_k \in \mathbb{C}$, we say $X_k=O_p(a_k)$, if for any $\epsilon>0$, there exists a finite $c$ and finite $k_0$, such that for any $k \geq k_0$, $P\left(|X_k|/|a_k| > c \right) \leq \epsilon$. Index $p$ in $O_p(\cdot)$ indicates ``in probability''. Similarly, $X_k=o_p(a_k)$ is defined. We write  $X_k=\Omega_p(a_k)$, if  $X_k=O_p(a_k)$ but  $X_k\neq o_p(a_k)$.  For a sequence of random matrices $\left(\mathsf{X}_k\right)_{k=1}^{\infty}$ and constants $\left(a_k\right)_{k=1}^{\infty}$, where $\mathsf{X}_k \in \mathbb{C}^{n,m}$ and $a_k \in \mathbb{C}$, we say $\mathsf X_k=O_p(a_k)$, if  $\norm{\mathsf X_k}_{\infty}=O_p(a_{k})$. Correspondingly, $\mathsf X_k=o_p(a_k)$ is defined.
Finally, for deterministic matrices, $O(\cdot)$ and $o(\cdot)$ are defined similarly.

For a sequence of random matrices $\left(\mathsf{A}_k\right)_{k=1}^{\infty}$, we say $\mathsf{A}_K\ra\mathsf{B}$ in probability with convergence rate $\upsilon>0$, if
\begin{IEEEeqnarray}{c}
\norm{\mathsf{A}_k - \mathsf B}_{\infty} = O_p\left(k^{-\upsilon}\right).
\end{IEEEeqnarray}

We say $f(x)\colon\mathbb{R}\rightarrow \mathbb{R}$ is asymptotically lower bounded by 
$g(x)\colon\mathbb{R}\rightarrow \mathbb{R}$, and write  $f(x) \geq g(x)$, if 
\begin{IEEEeqnarray}{c}
\lim \limits_{x \ra \infty} (f(x)-g(x)) \geq 0.
\label{eq:example}
\end{IEEEeqnarray}

Finally, for a discrete-time channel $\vc X\mapsto \vc Y$, where $\vc X, \vc Y\in \mathbb C^{m}$, with the average power constraint $\mathbb E ||\vc X||^2 \leq m\mathcal P$ and capacity $\mathcal{C}(\mathcal P)$, we say there are $r>0$ complex signal DoFs in the channel if the capacity pre-log is $r/m$, \ie\ $\lim\limits_{\mathcal P\ra\infty}\mathcal{C}(\mathcal P)/\log(\mathcal P)=r/m$.


\section{Split-step Fourier  model} \label{sec:model}

In this section, we consider a modified version of  SSFM introduced in \cite{yousefi2016cap}. Here, the nonlinearity and noise steps are combined into one step, so that the influence of the additive amplified spontaneous emission (ASE) noise can be seen as phase noise. In what follows, SSFM refers to the modified SSFM.


\subsection{Continuous-time model}

Denote the complex envelope of the optical signal at distance $z$ and time $t$ by  $Q(t, z)$. The propagation of the signal in single-mode optical fiber with distributed amplification is governed by the stochastic nonlinear Schr\"odinger (NLS) equation \cite[Eq.~2]{yousefi2016cap}
\begin{IEEEeqnarray}{c}
\frac{\partial{Q}}{\partial{z}}=L_L(Q)+L_N(Q)+N(t,z). \label{eq:schrodinger}
\end{IEEEeqnarray}
Here, $L_L$ is the linear operator 
\begin{IEEEeqnarray}{c}
 L_L(Q) = -j\frac{\beta_2}{2} \frac{\partial^2 Q}{\partial t^2}-\frac{1}{2}\alpha_r \convolution Q(t,z),
\label{eq:L-L}
\end{IEEEeqnarray}
where $\beta_2$ is the second-order chromatic dispersion coefficient,  $\alpha_r(t)$ is the residual attenuation coefficient (remained after imperfect amplification), $\convolution$ is convolution resulting from the dependence of the attenuation coefficient with frequency, and $j=\sqrt{-1}$.
The operator 
\begin{IEEEeqnarray}{rCl}
L_N(Q)=j\gamma \abs{Q}^2 Q
\label{eq:L-N}
\end{IEEEeqnarray}
represents the Kerr nonlinearity, where $\gamma$ is the nonlinearity parameter. 
Finally, $N(t,z)$ is zero-mean circularly-symmetric complex Gaussian noise process with covariance matrix %
\begin{IEEEeqnarray}{c}
\mathbb{E}\left[ N(t,z)N^*(t',z')\right] =\tilde{\sigma}^2 \delta_{\mathcal{B}}(t-t')\delta(z-z'),
\label{eq:N-cor}
\end{IEEEeqnarray}
where $\delta_{\mathcal{B}}(x) =\mathcal{B}\sinc (\mathcal{B}x)$, $\sinc (x)= \sin(\pi x)/(\pi x)$, $\mathcal{B}$ is the noise bandwidth,  
$\delta(\cdot)$ is the Dirac delta function, and $\tilde\sigma$ is the power spectral density of the ASE noise. Denote $\sigma\triangleq\tilde{\sigma}\sqrt{\mathcal{B}}$.

The capacity results obtained in this paper are expected to hold for more general linear and nonlinear operators $L_L$ and $L_N$ that take into account other forms of dispersion and nonlinearity. However, we restrict the analysis to \eqref{eq:L-L} and \eqref{eq:L-N}.


\subsection{Discrete-time SSFM model} \label{sec:ssfmDef}
Discretize a fiber of length $\mathcal{L}$ into $K$ segments of length $\varepsilon=\mathcal{L}/K$ in distance, and $Q(t,\cdot)$ to a vector of length $n$ with step size $\Delta_t$ in time. 
Let $\vc{V}_i\in\mathbb{C}^n$ be the input of the spatial segment $i$, where $\vc{V}_1=\vc{X}$ is the channel input and $\vc{V}_{K+1}=\vc{Y}_K$ is the channel output.

In segment $i$ of the modified SSFM, the following steps are performed \cite{yousefi2016cap}.

\begin{itemize}[leftmargin=*,wide]
\item[\textit{a})] \emph{Modified nonlinear step}: In this step  \eqref{eq:schrodinger} is solved analytically with $L_L=0$. Let $\vc{V}_i$  and $\vc{U}_i$ be the input and output in this step, and $M\ra\infty$ a large integer.
The channel $\vc V_i\mapsto \vc U_i$ is memoryless with the input output relation 
\cite[Eq. 5]{yousefi2016cap}
\begin{IEEEeqnarray}{c}
U_{i,\ell}= \left(V_{i,\ell}+W_{i,\ell}(M)\right)  e^{j\Phi_{i,\ell}},
\end{IEEEeqnarray}
where 
$U_{i,\ell}$ and $V_{i,\ell}$ are entries of $\vc U_i$ and $\vc V_i$ defined based on \eqref{eq:indexing}, and  $\bigl(W_{i,\ell}(m)\bigr)_{i,\ell}$ is a sequence of discrete-time Wiener processes in $m$ with auto-correlation function
\begin{IEEEeqnarray}{rCl}
{\sf E}\Bigl\{ W_{i,\ell}(m) W_{i',\ell'}(m') \Bigr\}=\sigma^2 \mu  \: \delta_{ii'}\delta_{\ell\ell'}\min(m,m'),
\IEEEeqnarraynumspace
\end{IEEEeqnarray}
for all $i,i'\in [K]$, $\ell,\ell'\in [n] $ and $m, m'\in[M]$, 
where  $\mu=\varepsilon/M$ and $\delta_{ij}$ is the  Kronecker delta function.  The  nonlinear phase is
\begin{IEEEeqnarray}{c}
    \Phi_{i,\ell} \triangleq \gamma \mu\sum \limits_{r=1}^{M} \Big|V_{i,\ell}+W_{i,\ell}(r) \Big|^2.
    \label{eq:Phi-ij}
\end{IEEEeqnarray}
Denote $\vc{\bar{Z}}_i\triangleq \vc{W}_{i}(M)$. 

\item[\textit{b})] \emph{Linear step}: In this step \eqref{eq:schrodinger} is solved analytically with $L_N=0$ and $N(t,z)=0$. If $\vc{U}_i$ is the input and $\vc{V}_{i+1}$ is the output of the linear step, the map  $\vc{U}_i \mapsto \vc{V}_{i+1}$ is
\begin{IEEEeqnarray*}{c}
\vc{V}_{i+1}=\mathrm{D}_K  \vc{U}_i, 
\end{IEEEeqnarray*}
where $\mathrm{D}_K$ is the (deterministic) dispersion matrix
\begin{IEEEeqnarray}{c}
\mathrm{D}_K=\mathrm{F}^{-1} \diag\left (\left( e^{a_{\ell}+jb_{\ell}}\right)_{\ell=1}^{n}\right)\mathrm{F}, \label{def:DK}
\end{IEEEeqnarray}  
where $\mathrm{F}$ is the discrete Fourier transform (DFT) matrix. Further, $a_{\ell}=-\mathcal{L}\alpha_{\ell}/(2K)$ where $(\alpha_{\ell})_{\ell}$ is a discretization of the loss coefficient $\alpha_r(f)$ in frequency $f$, and 
\begin{IEEEeqnarray}{c}
b_{\ell}=
-\frac{\mathcal{L}  \beta_2 (2\pi)^2 }{2K(n\Delta_t)^2} \times
\begin{cases}
\IEEEstrut
(\ell-1)^2,    & 1 \leq \ell \leq \frac{n}{2},\\
(n-\ell+1)^2 , & \frac{n}{2}+1 \leq \ell \leq n.
\IEEEstrut
\end{cases} \label{def:dispersions}
\end{IEEEeqnarray}%
\end{itemize}

\begin{remark}
The input dimension $n$ is fixed, and should not be confused with the block or codeword length that tends to $\infty$.
\end{remark}

\subsection{Transition from the continuous- to discrete-time model}
The NLS equation  \eqref{eq:schrodinger} defines a continuous-time channel from $Q(t,0)$ at the input of the fiber to $Q(t, \mathcal L)$ at the output of the fiber. Let $\mathcal{B}(z)$ denote the signal bandwidth at distance $z$.
To discretize the channel, we need to sample the input signal at $\mathcal{B}(0)$ and the output signal at $\mathcal{B}(\mathcal L)$. 
Due to Kerr nonlinearity, $\mathcal{B}(\mathcal L)$ is generally signal-dependent and may not equal to $\mathcal{B}(0)$. The relation between $\mathcal{B}(\mathcal L)$ and $\mathcal{B}(0)$ is an important open question. 

As a consequence, the continuous-time model \eqref{eq:schrodinger}  cannot be  discretized in a one-to-one manner as in linear channels by sampling the input output signals at the input bandwidth. In this paper, we do not include channel filters or a bandwidth constraint in the model, and let  $\mathcal{B} \eqdef \mathcal{B}(0)\ra\infty$.
The derivative operator in time can then be approximated using the discrete Fourier transform with an error that tends to zero as the step size in time $\Delta_t=1/\mathcal{B}\ra 0$. If $\snr =K^{\delta}$ and $\delta<1$, as $K\ra\infty$  operator splitting in distance yields a discretization of \eqref{eq:schrodinger} with a vanishing error as 
$\mathcal{B}\ra\infty$ and $K\ra\infty$. Finally, we consider a potentially sub-optimal discretization where the output signal is sampled at $\mathcal{B}$. This corresponds to a receiver that  ignores some of the samples that potentially carry information, and results in a discrete-time model in which the input output vectors have the same dimension.
Lower bounds obtained for this potentially sub-optimal receiver hold for better receivers as well.

\subsection{Limitations of the discrete-time model}
\label{sec:limitations}

The discrete-time model considered in this paper has a number of limitations. 
First, it considers signals with infinite bandwidth and does not account for a bandwidth constraint introduced by inline filters or receiver. Second, a receiver that samples the output signal at the input bandwidth ignores potentially useful samples. Third, the spectral efficiency (in bits/s/Hz) of the continuous-time model may not equal to the capacity (in bits/s) of  the discrete-time model. This is because the spectral broadening factor may 
increase with launch power due to nonlinearity \cite[Sec. VIII]{yousefi2011opc}, \cite{kramer2018autocorrelation}. 

In this paper, we do not derive a rigorous one-to-one discretization of the continuous-time model \eqref{eq:schrodinger}.
We consider the discrete-time SSFM model, whose capacity may be different from that of \eqref{eq:schrodinger}.

\section{A Capacity Lower Bound}
\label{sec:CapacityResults}

The capacity of the SSFM model as a function of the signal-to-noise ratio $\snr=\mathcal{P}/(\sigma^2 \mathcal{L})$, where  $\mathcal{P}$ is the average signal power and $\sigma^2 \mathcal{L}$ is total noise power, and the
number of spatial segments $K$ is
\begin{IEEEeqnarray}{c}
    \mathcal{C}(\snr, K){=}\max\limits_{p_{\vc X}(\vc x)} \Bigl\{
        \frac{1}{n}I(\vc{X};\vc{Y}_K): \:\frac{1}{n}\mathbb{E}\left[\norm{\vc{X}}^2 \right]\leq\mathcal{P} \Bigr\},
        \:\label{def:SSFMCapacity}
\IEEEeqnarraynumspace
\end{IEEEeqnarray}
where $\vc{X}\in\mathbb C^n$ and $\vc{Y}_K\in\mathbb C^n$ are the channel input and output, and 
$I(\vc{X};\vc{Y}_K)$ is the mutual information measured in bits/2D.

The capacity of the SSFM model when $K\ra\infty$ independently of \snr\ is
\begin{IEEEeqnarray}{c}
\mathcal{C}(\snr) \eqdef \lim\limits_{K\ra\infty}  \mathcal{C}(\snr, K).
\IEEEeqnarraynumspace
\end{IEEEeqnarray}
In this case, the asymptotic capacity corresponds to limits $\lim_{\snr\ra\infty}\lim_{K\ra\infty}$ with that order.
We also study the capacity when  $K$ and \snr\ go to infinity as $K=\sqrt[\delta]{\snr}$, $\snr\ra\infty$, in which case the capacity is  $\mathcal{C}(\snr, \sqrt[\delta]{\snr})$.

\begin{remark}
Since noise power is fixed by the channel, we express the capacity as a function of \snr\ instead of launch 
power. However, this should not imply that the capacity of the nonlinear channel, which is a two-dimensional function of signal and noise powers, is a one-dimensional function of \snr\ \cite[Sec. VII]{yousefi2011opc}.
\end{remark}

The main result of this paper is Theorem~\ref{th:mainThSSFM} stating that for sufficiently large number 
of segments $K$, rate $\frac{1}{2}\log_2(\snr)-\frac{1}{2}$ is achievable at high SNRs in the continuous-space lossless model.

\begin{figure}[t]
\centering
\includestandalone[scale=0.45]{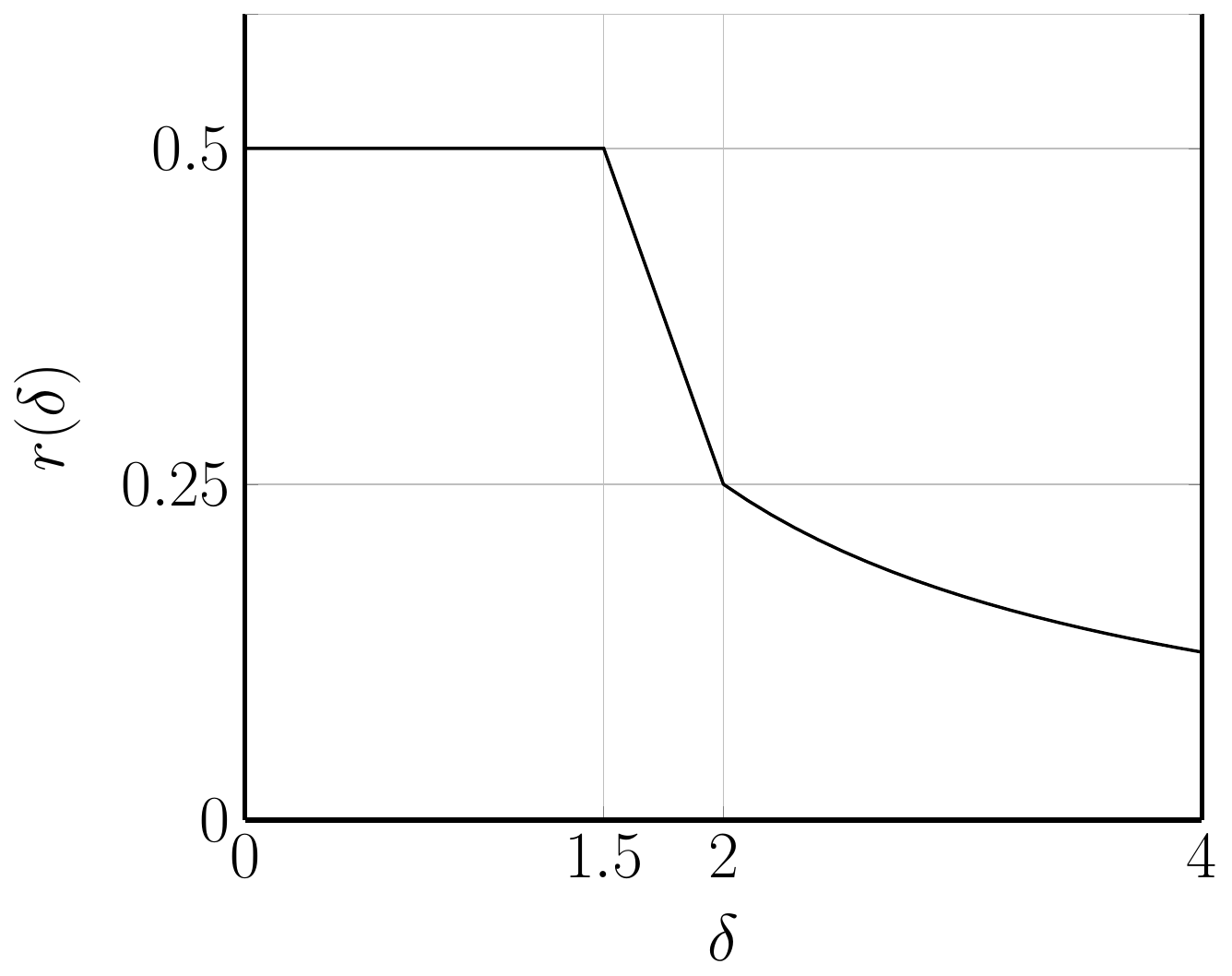}
\caption{Pre-log in the capacity lower bound as a function of the growth rate of the number of segments with power.}
 \label{fig:preLog}
\end{figure}

\begin{theorem} \label{th:mainThSSFM} 
The capacity of the SSFM channel when $K\ra\infty$ independently of \snr\ is lower bounded as 
\begin{IEEEeqnarray}{c}
     \mathcal{C}(\snr)\geq \frac{1}{2}\log_2\left(1+ \snr\right) + \frac{1}{2}\log_2 a+o(1),
     \label{eq:bound-a}
\end{IEEEeqnarray}
where the term $o(1)$ tends to zero with  $\snr\ra\infty$. For the lossless fiber, $a=1/2$.

The capacity of the SSFM channel when $K= \sqrt[\delta]{\snr}$, $\delta>0$, $\gamma\neq 0$, satisfies
\begin{IEEEeqnarray}{rCl}
\mathcal{C}\left(\snr,\sqrt[\delta]{\snr}\right) \geq r(\delta)\log_2(1+a\,\snr)+O(1), 
\quad\label{eq:capacity-a}
\end{IEEEeqnarray}
where $O(1)$ is bounded as $\snr \ra\infty$ and   the pre-log $r(\delta)$ is 
\begin{IEEEeqnarray}{rCl}
r(\delta) = 
\begin{cases}
\IEEEstrut
\frac{1}{2},&  0 < \delta \leq \frac{3}{2}, \\
\frac{3-\delta}{2\delta}, &  \frac{3}{2} \leq \delta \leq 2,\\
\frac{1}{2\delta}, &  2 \leq \delta<\infty.
\IEEEstrut
\end{cases}
\label{eq:r-theta}
\end{IEEEeqnarray}
Here, $a = \zeta e^{2\zeta } /(e^{2\zeta}-1)$, in which $\zeta=-\bar{\alpha}\mathcal{L}/2$, where $\bar{\alpha}$ is the average fiber loss in frequency. If $\delta\leq 3/2$, then $O(\cdot)$ can be replaced with $o(\cdot)$ in \eqref{eq:capacity-a}.
\end{theorem}

In practice the launch power is finite  and $K$ can be chosen to be arbitrarily large. 
Theorem \ref{th:mainThSSFM} indicates that rates given by \eqref{eq:bound-a} are achievable. 
In fiber-optic simulations based on SSFM,
choosing sufficiently large number of segments, the
channel capacity (bits/s) is between the lower bound in \eqref{eq:bound-a} and  upper bound \eqref{eq:ub}.

Theorem~\ref{th:mainThSSFM} indicates that the number of signal DoFs is at least half of the input dimension in the continuous-space model. In contrast, 
there is only one signal DoF in the discrete-space SSFM model where $\mathrm D_K \eqdef \mathrm D$ is independent of $K$.
The lower bound \eqref{eq:mainThSSFM} can be compared with the asymptotic capacity of the discrete-space SSFM model \eqref{eq:asym-cap}, where there is only one DoF, and the pre-log is $1/2n$ in the lossless model. 

The pre-log $1/2\delta\leq 1/8$ for $\delta \geq 4$ (part of the third branch in \eqref{eq:r-theta}) can be obtained from
\cite{kamran2015bound} as well. The maximum pre-log is however $1/2$.


\section{Proof of Theorem~\ref{th:mainThSSFM} and Related Capacity Theorems}
\label{sec:CapacityResultsExt}

In this section, we outline the proof of  Theorem~\ref{th:mainThSSFM}. In addition, we present a number of capacity theorems for models
related to the SSFM. Most of the proofs appear in Section~\ref{sec:proof}.

It is shown in \cite[Sec. V]{yousefi2016cap} that the SSFM channel when $K$ is fixed and the input distribution escapes to infinity with the average input power  tends to a linear fading channel.
To study the capacity of the SSFM channel when $K\ra\infty$, we first study this fading model in Section~\ref{sec:fading-K=infinity}, obtained by replacing nonlinearity with uniform \iid\ phase noise in segments of SSFM.  We show that if $K\ra\infty$, this fading channel tends to a diagonal phase noise channel in probability with capacity pre-log $1/2$.

Later in Section~\ref{sec:ssfm} we show that the limit of the  SSFM channel when the input distribution escapes to infinity with $\snr$ and there are  sufficiently large number of segments is  a continuous-space fading channel. 
Combining these two results, we obtain the pre-log $1/2$ in Theorem~\ref{th:mainThSSFM}. The cases where pre-log is less than $1/2$ are obtained similarly, with further analysis.

Note that, since we consider a continuous-space SSFM model, the proof in \cite{yousefi2016cap} showing that the SSFM channel tends to a fading channel with $\snr\ra\infty$ cannot be used here, because $\mathrm{D}_k$ does not depend on $K$ in the discrete-space model in \cite{yousefi2016cap}.


\subsection{Capacity of the continuous-space fading model} 
\label{sec:fading-K=infinity}

We begin by considering the following \emph{fading channel}, introduced in \cite{yousefi2016cap}:

\begin{IEEEeqnarray}{c}
 \mathbf{Y}_K=\mathsf{M}_K  \mathbf{X}+\mathbf{Z}_K, \label{def:fadingChannel}
\end{IEEEeqnarray}
where the random matrix $\mathsf{M}_K$ (independent of input), representing a multiplicative noise, is independent of $\vc X$ and given by \cite[Eq. 11]{yousefi2016cap}
\begin{IEEEeqnarray}{c}
\mathsf{M}_K= \prod \limits_{i=1}^{K} \left( \mathrm{D}_K  \mathsf{R}(\boldsymbol{\theta}_i)\right), \label{def:matrixMk}
\end{IEEEeqnarray}
where $\mathsf R(\cdot)$ is defined in \eqref{def:diagMat} and $\boldsymbol{\theta_i} \widesim{\text{i.i.d.}} \mathcal{U}^n(0,2\pi)$. The expression for the additive noise  is 
\begin{IEEEeqnarray}{c}
\vc{Z}_K= \sum \limits_{i=1}^K \left(\prod \limits_{r=i}^{K} \left( \mathrm{D}_K  \mathsf{R}(\boldsymbol{\theta}_r)\right)\right) \vc{\bar{Z}}_i,
\label{def:Noisek}
\end{IEEEeqnarray}
where 
$ \bar{\vc{Z}}_i \sim \mathcal{N}_{\mathbb{C}}\left(0, \sigma^2 \varepsilon\mathrm{I}_n\right)$.
In the lossless case, \eqref{def:Noisek} is simplified and $ \vc{Z}_K \sim \mathcal{N}_{\mathbb{C}}\left(0, \sigma^2 \mathcal{L}\mathrm{I}_n\right)$. In this paper, a fading channel is any model of the form \eqref{def:fadingChannel}, with multiplicative and additive noise, where $(\mathsf{M}_K, \mathbf{Z}_K)$ has arbitrary distribution and is independent of $\vc X$.


A special case of the fading channel \eqref{def:fadingChannel} is the non-coherent memoryless phase noise channel \cite{Lapidoth02}
\begin{IEEEeqnarray}{c}
\label{def:phaseNoiseChannel}
    Y_{\ell}=e^{j\theta_{\ell}} X_{\ell}+Z_{\ell},\quad {\ell}=1,2,\ldots, n,
\end{IEEEeqnarray}
where $\theta_{\ell} \sim \mathcal{U}(0,2\pi)$ and  $Z_{\ell} \sim \mathcal{N}_{\mathbb{C}}\left(0,\sigma_0^2 \right)$. The capacity of this channel, denoted by $\mathcal C_{P}$, is \cite[Eq.~23]{Lapidoth02}:
\begin{IEEEeqnarray}{c}
     \mathcal{C}_{P}(\snr)=\frac{1}{2}\log_2\left(1+\frac{1}{2}\snr\right) +o(1), \label{eq:capcityPhaseNoise}
\end{IEEEeqnarray}
where $\snr=\mathcal{P}/\sigma_0^2$ and the term $o(1)$ tends to zero with $\snr\ra\infty$. 

\begin{definition}
For a dispersion matrix of the form \eqref{def:DK}, we define the ``total dispersion values''  $d_{\ell}=Kb_{\ell} $, $\ell \in [n]$. We say dispersion is finite if  $b_{\ell}$ is  given by \eqref{def:dispersions}, for which $d_{\ell}<\infty$ independently of $K$, and infinite if $b_{\ell}$ does not depend on $K$, for which $d_{\ell} \rightarrow \infty$ as $K\rightarrow\infty$, $b_{\ell}\neq 0$. Similarly, we define the ``total loss values'' $\zeta_{\ell}=Ka_{\ell}$, $\ell \in [n]$.
\end{definition}

Denote the average values of the total loss and dispersion in frequency by
\begin{IEEEeqnarray}{c}
    \zeta \triangleq \frac{1}{n} \sum \limits_{\ell=1}^n \zeta_{\ell},\quad 
    d \triangleq \frac{1}{n} \sum \limits_{\ell=1}^n d_{\ell}.
    \label{def:averageLoss}
\end{IEEEeqnarray}
Alternatively, $\zeta =-\bar{\alpha}\mathcal L/2$, where $\bar\alpha \triangleq (1/n) \sum_{\ell=1}^n \alpha_{\ell}$ is the average
fiber loss.
If $n\ra\infty$, $\bar{\alpha} =\lim_{F\ra\infty}\frac{1}{F}\int_{-F/2}^{F/2}\alpha_r(f) df$.

For realistic fiber, $b_{\ell}$ is given by \eqref{def:dispersions}, and dispersion is finite due to factor $1/K$ 
in \eqref{def:dispersions}. In this case, the effect of dispersion locally in a small segment is infinitesimal, and $\lim_{K\ra\infty}\mathrm D_K= I_n$.
However, we also consider models where $b_{\ell}$ is independent of $K$; here the dispersion matrix in each small segment is  fixed $\mathrm D_K \eqdef \mathrm D$, and the  total dispersion value is infinite for $K\ra\infty$. 
The models with infinite dispersion or nonlinearity are not realistic, but help to understand the capacity as $\beta_2$ or $\gamma$ tend to infinity. They distinguish  models with fixed and variable $D_K$.


Below, we lower bound the capacity of the fading channel \eqref{def:fadingChannel} with finite and infinite dispersion.

\subsubsection{Finite Dispersion} 

A key result of this paper is Lemma~\ref{lem:asymDistK} stating that as $K\ra\infty$, the random matrix $\mathsf{M}_K$ in \eqref{def:fadingChannel} tends to a diagonal matrix with independent phase noise components. 

For a random vector $\boldsymbol{\theta}$, define
\begin{IEEEeqnarray}{c}
\mathsf{L}(\boldsymbol{\theta})\triangleq\mathsf{R}(\boldsymbol{\theta}) \mathrm{C}_1 \mathsf{R}(\boldsymbol{\theta})^{-1},
\label{eq:Ltheta}
\end{IEEEeqnarray}
where
$\mathrm{C}_1  \triangleq  \mathrm{F}^{-1}  
\diag\bigl( \left(\zeta_{\ell} + jd_{\ell}\right)_{\ell=1}^n\bigr)
\mathrm{F}$ and $\bar{\mathsf{L}}\triangleq \mathbb{E}_{\boldsymbol{\theta}} \left[\mathsf{L}(\boldsymbol{\theta}) \right]$.

\begin{lemma} 
\label{lem:asymDistK}  
The random matrix $\mathsf{M}_K$ has the expansion 
\begin{IEEEeqnarray}{rcl}
 \mathsf{M}_K &\stackrel{d}{=}&\left(e^{\zeta+jd}\mathrm{I}_n+\frac{1}{K}\sum \limits_{i=1}^{K}\mathsf{L}(\boldsymbol{\theta}_i)-\bar{\mathsf{L}}\right)\mathsf{R}\left(\boldsymbol{\theta}\right)+O_p\left(\frac{1}{K}\right),\nonumber\\
   &\stackrel{d}{=}&e^{\zeta}\mathsf{R}\left(\boldsymbol{\theta}\right)+O_p\left(\frac{1}{\sqrt{K}}\right),\label{eq:Mk-Complete}
\end{IEEEeqnarray}
where  $\boldsymbol{\theta}$ and $\boldsymbol{\theta}_i$ are drawn i.i.d. from $\mathcal{U}^n(0,2\pi)$.
In particular,
\begin{IEEEeqnarray}{c}
\lim \limits_{K \rightarrow   \infty} \mathsf{M}_{K}
\stackrel{d}{=}
e^{\zeta} \mathsf{R}(\boldsymbol{\theta}).
\label{eq:lim-Mk}
\end{IEEEeqnarray}
\end{lemma}

\begin{proof}
The proof is given in  Section~\ref{pr:asymDistK}. The first equality is shown by algebra. The second equation is obtained by applying a concentration inequality for weak law of large numbers, or central limit theorem.  
\end{proof}

Substituting the limit of $\mathsf M_K$ in \eqref{def:matrixMk} into \eqref{def:Noisek}, we also obtain
 $\lim_{K\ra\infty}\vc{Z}_K=\vc{Z}$, where $\vc{Z} \sim \mathcal{N}_{\mathbb{C}}(0,\sigma'\mathrm{I}_n)$, $\sigma'=
\sigma^2 \mathcal{L} (e^{2\zeta} - 1)/2\zeta$. The following lemma follows.
\begin{lemma} 
\label{lem:asymFadingK} 
As $K\ra\infty$, the fading channel \eqref{def:fadingChannel} with finite dispersion tends to a sequence of independent phase noise channels

\begin{IEEEeqnarray}{c}
\label{def:phaseNoiseChannels}
    Y_{\ell}=e^{\zeta+j\theta_{\ell}} X_{\ell}+Z_{\ell},\quad \ell=1,\ldots, n,
\end{IEEEeqnarray}
where $\theta_{\ell} \widesim{\text{i.i.d.}} \mathcal{U}(0,2\pi)$ and $Z_{\ell} \widesim{\text{i.i.d.}}   \mathcal{N}_{\mathbb{C}}\left(0,\eta \sigma^2 \mathcal{L} \right)$, in which
\begin{IEEEeqnarray}{c}
        \eta = \begin{cases} 
        \IEEEstrut
        1, &\textnormal{ \emph{for lossless channel},}\\
        \frac{e^{2\zeta}-1}{2\zeta}, &\textnormal{\emph{otherwise}.}
        \IEEEstrut
        \end{cases} \label{def:eta}
    \end{IEEEeqnarray}
\end{lemma}

The first capacity result in this paper is the following theorem showing that the pre-log of the capacity of the continuous-space fading channel with finite dispersion when  $K\ra\infty$ independently of $\snr$ is $1/2$, \ie\ there are $n$ real signal DoFs when the input dimension is $2n$. Let $\bar{\mathcal{C}}(\snr, K)$ be the capacity of \eqref{def:fadingChannel} and  $ \bar{\mathcal{C}}(\snr)=\lim_{K \rightarrow \infty } \bar{\mathcal{C}}(\snr, K)$.

\begin{theorem} 
\label{th:asymCapFadingFinDis} 
Capacity of fading channel \eqref{def:fadingChannel} with finite dispersion  and \snr\ satisfies
\begin{IEEEeqnarray}{c}
     \bar{\mathcal{C}}(\snr)= \frac{1}{2}\log_2\left(1+a \,\snr\right) +o(1),
     \label{eq:capacityFadingFinite}
\end{IEEEeqnarray}
where $\snr=\mathcal{P}/(\sigma^2 \mathcal{L})$ and $a$ is given in Theorem~\ref{th:mainThSSFM}. 
\end{theorem}

\begin{proof}
The result is obtained by combining Lemmas~\ref{lem:asymDistK} and \ref{lem:asymFadingK} and the capacity result \eqref{eq:capcityPhaseNoise}. 
\end{proof}

Note that \eqref{eq:capacityFadingFinite} holds for any \snr, if $K\rightarrow\infty$ independently of \snr. Next, we consider the capacity of the finite dispersion fading channel when $ K = \sqrt[\delta]{\snr}$,
$\delta>0$, $\snr \rightarrow \infty$. We require this model in the Section~\ref{sec:infinite-dips} when we study the SSFM channel. Substituting \eqref{eq:Mk-Complete} into \eqref{def:fadingChannel},  \eqref{def:phaseNoiseChannel} is modified to
\begin{IEEEeqnarray}{c}
\vc{Y}_K\stackrel{p}{\ra}  e^{\zeta}\mathsf{R}(\boldsymbol{\theta}) \vc{X}+\vc{Z} + \Delta \vc{X},\quad \textnormal{as~~~} K\ra\infty,
\end{IEEEeqnarray}

where $\Delta=O_p\left(K^{-1/2}\right)$ and hence, $\Delta\vc{X}=O_p\left(K^{(\delta-1)/2}\right)$. Here, $\Delta$ is a dense matrix, that is generally not diagonal, capturing ISI or intra-channel interactions.

If $\delta<1$, then $\Delta \vc{X} \stackrel{p}{\ra} (0,\ldots,0)$ and $\vc{Y}_K \stackrel{p}{\ra} e^{\zeta}\mathsf{R}(\boldsymbol{\theta}) \vc{X}+\vc{Z}$.
In this case, the capacity is provided by Theorem~\ref{th:asymCapFadingFinDis}.

If $\delta>1$, $\Delta \vc{X}$ grows with $K$ and constitutes the dominant stochastic impairment. The following theorem establishes a lower bound on the capacity in this case.

\begin{theorem} 
\label{th:fadingHighPowerdG1}
The capacity of the fading channel \eqref{def:fadingChannel} when $K= \sqrt[\delta]{\snr}$, $\delta>1$, is lower bounded as
\begin{IEEEeqnarray}{c}
 \bar{\mathcal{C}}\left(\snr, \sqrt[\delta]{\snr}\right) \geq\frac{1}{2 \delta} \log_2(1+\snr)+O(1),
\end{IEEEeqnarray}
where the $O(1)$ term is bounded as $\snr \ra \infty$.
\end{theorem}
\begin{proof}
See Section~\ref{pr:fadingHighPowerdG1}.
\end{proof}


\subsubsection{Infinite Dispersion}
\label{sec:infinite-dips}

We consider the SSFM model when $\mathrm{D}_K \eqdef \mathrm{D}$ is independent of $K$.
Lemma~\ref{lem:haarMeasure} below shows that, in the lossless infinite dispersion case,
as $K\rightarrow \infty$ the random matrix $\mathsf{M}_K$ in \eqref{def:fadingChannel} tends to a random unitary matrix.

\begin{lemma} \label{lem:haarMeasure}
Let $\nu$ be  an   equivalence   relation  on   the  set $\{1,2,\ldots, n\}$, and  
$\mathbb{U}_n(\nu)$  the smallest subgroup of block diagonal matrices in $\mathbb{U}_n$ that contains $\mathrm{D}$. Then, the distribution of $\mathsf{M}_{K}$ tends to the Haar measure on $\mathbb{U}_n(\nu)$ as $K\ra\infty$.
\end{lemma}

\begin{proof} 
See Section~\ref{pr:haarMeasure}.
\end{proof}

In the following assume that $\mathrm{D}$ is not a block diagonal matrix (of more than one block), \textit{i.e.} $\mathbb{U}_n(\nu)=\mathbb{U}_n$. 

Decompose the mutual information for the fading channel \eqref{def:fadingChannel} as:
\begin{IEEEeqnarray}{rcl}
    I\Big(\vc{X};\vc{Y}_K\Big) = I\Big(\vc{X};\norm{\vc{Y}_K}\Big)+I\left(\vc{X};\hat{\vc{Y}_K}\big|\norm{\vc{Y}_K}\right). \label{eq:mutualILossless1}
\end{IEEEeqnarray}
Lemma~\ref{lem:haarMeasure} and Theorem~\ref{th:independence}  in Appendix~\ref{sec:math} imply that the second term approaches zero as $K \rightarrow \infty$. The conditional PDF of the signal norm is
\begin{IEEEeqnarray}{rCl}
P_{\norm{\vc{Y}_K}\big|\vc{X}}\left(\norm{\vc{y}}\big|\vc{x}\right)&=&P_{\norm{\vc{Y}_K}\big|\norm{\vc{X}}}\left(\norm{\vc{y}}\big|\norm{\vc{x}}\right) \nonumber \\
&=&\frac{2\norm{\vc{y}}^n}{\sigma^2 \mathcal{L} \norm{\vc{x}}^{n-1}} \exp\left(-\frac{\norm{\vc{y}}^2+\norm{\vc{x}}^2}{\sigma^2 \mathcal{L}} \right) \nonumber
\\
&&\times
\mathrm{I}_{n-1}\left(\frac{2\norm{\vc{x}}\norm{\vc{y}}}{\sigma^2 \mathcal{L}}\right),
 \label{eq:conditionalFadingOutputNorm}\
\end{IEEEeqnarray} 
where $I_m(\cdot)$ is the modified Bessel function of the first kind with order $m$.
Equations \eqref{eq:mutualILossless1} and \eqref{eq:conditionalFadingOutputNorm}  yield the following theorem.

\begin{theorem} \label{th:infFadCap}
Suppose that $\mathrm D \eqdef \mathrm{D}_K$ is independent of $K$ and  is not a block diagonal matrix.
The continuous-space fading channel \eqref{def:fadingChannel} has one signal DoF with capacity
\begin{IEEEeqnarray}{rCl}
\bar{\mathcal{C}}(\snr)\
&=& \max\limits_{P_{\norm{\vc{X}}}(\norm{\vc{x}})}
\Bigl\{ \frac{1}{n} I(\norm{\vc{X}};\norm{\vc{Y}}) :
\frac{1}{n}\mathbb{E}\left[\norm{\vc{X}}^2 \right]\leq\mathcal{P} \Bigr\}
 \nonumber\\
&=&
\frac{1}{2n} \log_2\left(1+\snr \right)+O(1), \label{eq:capacityFiniteFading}
\IEEEeqnarraynumspace
\end{IEEEeqnarray}
where $P\left(\norm{\vc{y}}\big|\norm{\vc{x}}\right)$ is given by ~\eqref{eq:conditionalFadingOutputNorm}, $\snr=\mathcal{P}/\sigma^2 \mathcal{L}$ and the $O(1)$ term  is bounded as $\snr \ra\infty$. 

\end{theorem}

\begin{proof}
The first line follows from Lemma~\ref{lem:haarMeasure}, stating that if $K\ra\infty$, then $I\left(\vc{X};\hat{\vc{Y}}\big|\norm{\vc{Y}}\right)\ra 0$. 
The second line follows from the lower bound on the capacity of  the non-central chi-square channel derived in  \cite[Theorem~1]{Shev2018} and an upper bound derived by noting that $h(\|\vc{Y}\|) {\leq} \frac{1}{2} \log_2\left(1+\snr \right)+O(1)$ due to principle of maximum entropy and $h(\|\vc{Y}\|\big| \|\vc{X}\|)=O(1)$.
\end{proof}

\begin{remark}
In general, if $\mathbb{U}_n(\nu)$ is the smallest subgroup of block diagonal matrices containing $\mathrm{D}$, then the capacity of the continuous-space fading channel with infinite dispersion has $m(\nu)$ real signal DoFs, where $m(\nu)$ is the number of blocks of $\mathrm{D}$.
\end{remark}


\subsection{Capacity of the continuous-space SSFM model} 
\label{sec:ssfm}

In this section,  the capacity of the continuous-space SSFM model is investigated in the high power regime, as well as with infinite nonlinearity, and infinite dispersion.


\subsubsection{High power regime}
Let $K = \sqrt[\delta]{\snr}$, $\delta>0$, and assume that the input $\vc{X}$ escapes to infinity with \snr. We show that the limit of the discrete-space SSFM channel when $\snr\ra\infty$ is a diagonal model with phase noise. The convergence rate to this diagonal model depends on $\delta$. We derive a lower bound on the convergence rate, from which a lower bound on the capacity is established. 

Denote
\begin{IEEEeqnarray}{c}
\mathsf{S}_K \triangleq \diag\left( \Bigl( \exp\Bigl(j \sum \limits_{i=1}^K \Phi_{i,\ell}\Bigr) \Bigr)_{\ell=1}^n \right),
\end{IEEEeqnarray}
and
\begin{IEEEeqnarray}{c}
       \upsilon(\delta)= \begin{cases}
       \IEEEstrut
        5\delta/6,& 0 < \delta \leq 1,\\
        1-\delta/6,& 1 \leq \delta \leq 1.5 ,\\
        1.5-\delta/2,& 1.5 \leq \delta \leq 2 \;\; \\
  &\text{and i.i.d. $\vc X$}, \\
         0.5,& 2 \leq \delta \leq 3 \;\; \\
  & \text{and i.i.d. $\vc X$},\\
    0.5,& 3 \leq \delta.
    \IEEEstrut
\end{cases} \label{def:upsilonDelta}
\end{IEEEeqnarray}

\begin{lemma} 
\label{lem:mainThSSFM-ex} 

If the input distribution  is absolutely continuous  and escapes to infinity with $\snr$, the SSFM channel with $K=\sqrt[\delta]{\snr}$, $\delta>0$, $\gamma\neq 0$, tends to the following channel in distribution as $\snr\ra\infty$
\begin{IEEEeqnarray}{c}
\vc{Y}_K=\mathsf{M}_K \vc{X}+ \vc{Z},
\end{IEEEeqnarray}
where
\begin{IEEEeqnarray}{c}
    \mathsf{M}_K=e^{\zeta+jd} \mathsf{S}_K+O_p\left(K^{-\upsilon(\delta) + \epsilon'}\right),
\end{IEEEeqnarray}
and $\vc{Z}\sim \mathcal{N}_{\mathbb{C}}(0,\eta\sigma^2\mathcal{L}\mathrm{I}_n)$, where $\eta$ is defined in Lemma~\ref{lem:asymFadingK}. Further,       $\epsilon'>0$ can be made arbitrarily small with \snr.
\end{lemma}

\begin{proof}
See Section~\ref{pr:mainThSSFM-ex}.
\end{proof}

Theorem~\ref{th:mainThSSFM} will be proved using Lemma~\ref{lem:mainThSSFM-ex} in Section \ref{pr:mainThComp}. The term $O(1)$ in Theorem~\ref{th:mainThSSFM} for $\delta\geq 2$ is given in the proof of Theorem~\ref{th:fadingHighPowerdG1}.

\begin{remark} \label{rem:lowNoise}
For $0\leq\delta \leq 1.5$, the dominant term in $O_p\left(K^{-\upsilon(\delta)}\right)$ is the signal-noise mixing, and for $1.5 \leq \delta \leq 3$, is the  intra-channel interactions (between different indices). For $3 < \delta$ both effects are significant. 
\end{remark}


\subsubsection{Infinite nonlinearity} \label{sec:ssfmGamma}

It is commonly believed that nonlinearity is a distortion that reduces the capacity.
However, in this section we show that when $\gamma \ra \infty$, the continuous-space channel has
at least $n$ real signal DoFs.

\begin{theorem} \label{th:infGamma}
Capacity of the continuous-space SSFM model  satisfies
\begin{IEEEeqnarray}{c}
\lim \limits_{K \rightarrow \infty} \lim \limits_{\gamma  \rightarrow \infty}  \mathcal{C}(\snr, K)= \frac{1}{2}\log_2\left(1+a\,\snr\right)+o(1),
    \end{IEEEeqnarray}
where $a$ is given in Theorem~\ref{th:mainThSSFM} and $o(1)$ term tends to zero with $\snr\rightarrow \infty$.
\end{theorem}

\begin{proof} It is easy to verify that as $\gamma/K\ra \infty$, for any $i\in [K]$ and $\ell\in[n]$, $\Phi_{i,\ell} \ra \infty$, and consequently, $\mod\left(\Phi_{i,\ell} ,2\pi \right)\ra \mathcal{U}(0,2\pi)$, independent of $\Phi_{i',\ell'}$ in any other segment $l'$ or coordinate $i'$. Hence, the SSFM channel tends to a finite dispersion fading channel, and Theorem~\ref{th:asymCapFadingFinDis} yields the result.
\end{proof}


\subsubsection{Infinite dispersion} \label{sec:ssfmInfDisp}

The asymptotic capacity of the discrete-space SSFM channel where $\mathrm D_K \eqdef \mathrm D$ is independent of $K<\infty$ is given in \eqref{eq:asym-cap}. 
In this section, we show that this result holds under the same assumption $\mathrm D_K \eqdef \mathrm D$ for the continuous-space lossless SSFM channel as well. Note that in this case, as $K\ra\infty$ the dispersion is infinite.

\begin{theorem}\label{th:infDispSSFM} 
Consider the SSFM model when $\mathrm D_K \eqdef \mathrm D$ is independent of $K$ and is not a block diagonal matrix of more than one block. If $K = \sqrt[\delta]{\snr}$, $\delta>3$, $\gamma\neq 0$, then
\begin{IEEEeqnarray}{c}
    \mathcal{C}\left(\snr,\sqrt[\delta]{\snr}\right) = \frac{1}{2n} \log_2\left(1+\snr \right)+O(1),
\end{IEEEeqnarray}
where $O(1)$ term is bounded as $\snr \ra \infty$. 
\end{theorem}

\begin{proof}
Similar to the analysis in \cite{yousefi2016cap}, it can be shown that for input distributions that escape to infinity with $\mathcal{P}=(\sigma^2 \mathcal{L})K^{\delta}$,  $\delta>3$, $V_{i,\ell} =\Omega_p\left(K^{\delta}\right)$ for all $i \in [K]$ and $\ell \in [n]$. Furthermore, in the proof of Lemma~\ref{lem:mainThSSFM-ex} it is shown that in this case $\mod\left(\Phi_{i,\ell},2\pi\right) \ra  \mathcal{U}(0,2\pi)$.  Hence, infinite dispersion SSFM channel tends to the infinite dispersion fading channel. These two steps can be proved alternatively using induction on the output of segment $i \in [K]$ and using Lemma~\ref{lem:haarMeasure} for sufficiently large $K$.  The result then follows from Theorem~\ref{th:infFadCap}.
\end{proof}

\section{Proofs} \label{sec:proof}

\subsection{Proof of Lemma~\ref{lem:asymDistK}}
\label{pr:asymDistK}

For matrices $A$ and $B$ define the commutator $[\mathrm{A},\mathrm{B}] = \mathrm{A}  \mathrm{B}  \mathrm{A}^{-1}  \mathrm{B}^{-1}$. It can be verified with algebraic manipulations that $\mathsf{M}_K$ can be written as
\begin{IEEEeqnarray}{rcl}
\mathsf{M}_K&=&\prod \limits_{i=1}^K \mathrm{D}_K  \mathsf{R}(\boldsymbol{\theta}_i)\nonumber\\
&=&\left\{\prod \limits_{i=2}^{K}  \mathrm{D}_K   \left[\prod \limits_{\ell=i}^K \mathsf{R}(\boldsymbol{\theta}_{\ell}),\mathrm{D}_K\right]\right\}  \mathrm{D}_K   \prod \limits_{\ell=1}^K \mathsf{R}(\boldsymbol{\theta}_{\ell})\nonumber\\
&=&\left\{\prod \limits_{i=2}^{K}  \mathrm{D}_K   \left[ \mathsf{R}\left(\sum \limits_{\ell=i}^K \boldsymbol{\theta}_{\ell}\right),\mathrm{D}_K\right]\right\} \mathrm{D}_K    \mathsf{R}\left(\sum \limits_{\ell=1}^K \boldsymbol{\theta}_{\ell}\right).
\end{IEEEeqnarray}

Since the joint distribution of $\Bigl(\sum_{\ell=i}^K\boldsymbol{\theta}_{\ell}\Bigr)_{i=1}^K$ and $\left(\boldsymbol{\theta}_i\right)_{i=1}^K$  are the same
\begin{IEEEeqnarray}{c}
\mathsf{M}_K\stackrel{d}{=} \left\{\prod \limits_{i=2}^{K} \left( \mathrm{D}_K   \left[ \mathsf{R}\left(\boldsymbol{\theta}_{i} \right),\mathrm{D}_K\right]\right) \right\} \mathrm{D}_K    \mathsf{R}\left(\boldsymbol{\theta}_1\right). \label{eq:MKreWritten}
\end{IEEEeqnarray}
Let 
\begin{IEEEeqnarray}{c}
\mathrm C_m\triangleq\frac{1}{m!}\mathrm{F}^{-1}  
\diag\Bigl( \bigl( (\zeta_{\ell}+jd_{\ell})^m \bigr)_{\ell=1}^n\Bigr)\mathrm{F},
\label{eq:C-i}
\end{IEEEeqnarray}
and $\bar{\mathrm{C}}_m \triangleq  (-1)^m \mathrm{C}_m$. 
Expand $\mathrm{D}_K$ and $\mathrm{D}_K^{-1}$ in $1/K$ using the Taylor's theorem 
\begin{IEEEeqnarray}{rcl}
\mathrm{D}_K&=&\mathrm{I}_n+\frac{1}{K}\mathrm{C}_1+\frac{1}{K^2}\mathrm{C}_2+O\left(\frac{1}{K^3}\right), \label{eq:Dk}\\
\mathrm{D}_K^{-1}&=&\mathrm{I}_n+\frac{1}{K}\bar{\mathrm{C}}_1+\frac{1}{K^2}\bar{\mathrm{C}}_2+O\left(\frac{1}{K^3}\right). \nonumber
\end{IEEEeqnarray}

A simple calculation shows that
\begin{IEEEeqnarray}{c}
 [\mathsf{R}(\boldsymbol{\theta}_i),D_K]=\mathrm{I}_n+\frac{\mathsf{L}(\boldsymbol{\theta}_i)+ \bar{\mathrm{C}}_1}{K}+O_p\left(\frac{1}{K^2}\right).
\end{IEEEeqnarray}
Using \eqref{eq:Dk} and $\mathrm{C}_1+\bar{\mathrm{C}}_1=0$,
\begin{IEEEeqnarray}{c}
\mathrm{D}_K   [\mathsf{R}(\boldsymbol{\theta}_i),D_K]=\mathrm{I}_n+\frac{\mathsf{L}(\boldsymbol{\theta}_i)}{K}+O_p\left(\frac{1}{K^2}\right).
\end{IEEEeqnarray}
Combining \eqref{eq:MKreWritten} and the above relation results in
\begin{IEEEeqnarray}{rcl}
  \mathsf{M}_K &=& \left\{\prod \limits_{i=2}^{K} \left( \mathrm{I}_n+\frac{\mathsf{L}(\boldsymbol{\theta}_i)}{K}+O_p\left(\frac{1}{K^2}\right) \right) \right\}  \mathrm{D}_K    \mathsf{R}\left(\boldsymbol{\theta}_1\right)\nonumber \\
  &=& \left\{\prod \limits_{i=2}^{K} \left( \mathrm{I}_n+\frac{\bar{\mathsf{L}}}{K}+\frac{\mathsf{L}(\boldsymbol{\theta}_i)-\bar{\mathsf{L}}}{K} \right) \right\}      \mathsf{R}\left(\boldsymbol{\theta}_1\right)+O_p\left(\frac{1}{K}\right)\nonumber \\
  &=& \left(\left\{\prod \limits_{i=2}^{K} \left( \mathrm{I}_n+\frac{\bar{\mathsf{L}}}{K} \right) \right\}+\frac{\sum \limits_{i=2}^K\left(\mathsf{L}(\boldsymbol{\theta}_i)-\bar{\mathsf{L}}\right)}{K}\right)      \mathsf{R}\left(\boldsymbol{\theta}_1\right)\nonumber\\
  &&+O_p\left(\frac{1}{K}\right)\nonumber\\
  &\stackrel{(a)}{=}& \left(e^{\bar{\mathsf{L}}}+\frac{1}{K}\sum \limits_{i=1}^{K}\mathsf{L}(\boldsymbol{\theta}_i)-\bar{\mathsf{L}}\right)  \mathsf{R}\left(\boldsymbol{\theta}_1\right)+O_p\left(\frac{1}{K}\right).
\end{IEEEeqnarray}
where $(a)$ is obtained using 
\begin{IEEEeqnarray}{c}
\left(1+\frac{\bar{\mathsf{L}}}{K}\right)^K = e^{\bar{L}}+O\left(\frac{1}{K}\right).
\end{IEEEeqnarray}

Finally, since $e^{\bar{\mathsf{L}}}=e^{(\zeta+jd)\mathrm{I}_n}=e^{\zeta+jd}\mathrm{I}_n$,
\begin{IEEEeqnarray*}{c}
 \mathsf{M}_K \stackrel{d}{=}\left(e^{\zeta+jd}\mathrm{I}_n+\frac{1}{K}\sum \limits_{i=1}^{K}\mathsf{L}(\boldsymbol{\theta}_i)-\bar{\mathsf{L}}\right)\mathsf{R}\left(\boldsymbol{\theta}\right)+O_p\left(\frac{1}{K}\right).
\end{IEEEeqnarray*}

\subsection{Proof of Theorem~\ref{th:fadingHighPowerdG1}} \label{pr:fadingHighPowerdG1}

In the following, we restrict the input to the class of  absolutely continuous random vectors, for which $h\left(\mathsf{M}_K \vc{X}\right)$
is a continuous function of $K$ with respect to the total variation distance \cite[Theorem~1]{GhourchianGohariAmini17}. First, we prove that:

\textit{i.}  If $\vc X$ is i.i.d, then
\begin{IEEEeqnarray}{rcl}
\frac{1}{n} I(\vc{X};\vc{Y}_K)&\geq& \frac{1}{2\delta} \log_2(\mathcal{P})+ h\left(|X_1|\right)-\mathbb{E}\left[\log_2\left(\norm{\mathbf{X}}_4\right)\right]\nonumber \\
&& +\frac{1}{2} \log_2\left(\frac{    e^{2\zeta-1}}{2\rho\pi} \right)+o(1),
\end{IEEEeqnarray} where the $o(1)$ term vanishes with $\mathcal{P}  \ra\infty$ and 
\begin{IEEEeqnarray}{c}
\rho \triangleq
\sqrt{\sum \limits_{r=1}^n \Bigg| \sum \limits_{s=1}^n (\zeta_s+jd_s)e^{-\frac{j2\pi rs}{n}}\Bigg|^4}.
\end{IEEEeqnarray} 
Next, using this general lower bound we obtain the following.

\textit{ii.} By choosing $\vc{X}\sim \mathcal{N}_{\mathbb{C}}(0,\mathcal{P}\mathrm{I}_n)$,  we have
\begin{IEEEeqnarray}{cl}
 \bar{\mathcal{C}}(\snr, \sqrt[\delta]{\snr}) \geq &\frac{1}{2 \delta} \log_2\left(1+\snr\right)+\frac{1}{2} \log_2\left(\frac{e^{2\zeta+1}}{\rho\pi\sqrt{8n} }\right)\nonumber \\
 &+o(1). \IEEEeqnarraynumspace \label{eq:fadingLower}
\end{IEEEeqnarray}

\textit{Part i:}
Define the matrix $\Delta \eqdef \frac{1}{K}\sum \limits_{m=1}^{K}\mathsf{L}(\boldsymbol{\theta}_m)-\bar{\mathsf{L}}$.
Considering \eqref{eq:Ltheta},  $\Delta_{rr}$ is deterministic and thus zero. If $r\neq \ell$,
\begin{IEEEeqnarray}{rcl}
\Delta_{\ell r}&=&
\frac{(C_1)_{\ell r}}{K}\sum\limits_{i=1}^{K}e^{j(\theta_{i,\ell}-\theta_{i,r})} \nonumber
\\
&\stackrel{d}{\ra}& \frac{ (C_1)_{\ell r}}{\sqrt{K}}T_{\ell r}
\end{IEEEeqnarray}
where $T_{\ell r}\sim\mathcal{N}_{\mathbb{C}}(0,1)$ and we used the central limit theorem.
This yields 
\begin{IEEEeqnarray}{c}
  \mathsf{M}_K \stackrel{d}{=}  e^{\zeta} \mathsf{R}\left(\boldsymbol{\theta}\right)+O_p\left(\frac{1}{\sqrt{K}}\right).
\end{IEEEeqnarray}
Hence,
\begin{IEEEeqnarray}{rcl}
h(\mathbf{Y}_K)&=&h(\mathsf{M}_K  \mathbf{X})+o(1)\nonumber\\
&=& n\log_2\left( e^{\zeta} \mathsf{R}(\boldsymbol{\theta})\mathbf{X}\right)+o(1) \nonumber\\
&=&n\log_2\left( 2\pi e^{2\zeta} \right)+ n h\left(|X_1|\right)+n\mathbb{E}\left[\log_2\left(|X_1|\right)\right]+o(1),\nonumber\\ \label{eq:outputEntropy}
\end{IEEEeqnarray}
where $o(1)$ term vanishes with $K \ra \infty$ and $\mathcal{P}/K\ra \infty$.

Next, we bound the conditional entropy part as:
\begin{IEEEeqnarray}{c}
h(\mathbf{Y}_K|\mathbf{X})\leq  \sum \limits_{\ell=1}^n h(Y_{K,\ell}|\mathbf{X}). \label{eq:cEntropySum}
\end{IEEEeqnarray}

The output $Y_{K,\ell}$ is equal in probability to 
\begin{IEEEeqnarray}{c}
 e^{\zeta+j(d+\theta_{\ell})}X_{\ell}+\frac{e^{j\theta_{\ell}}}{\sqrt{K}}\sum \limits_{r\neq \ell}(\mathrm{C}_{1})_{\ell r} T_{\ell r} X_r+\sqrt{\eta}~Z_{\ell}+o_p(1),\nonumber\\ \end{IEEEeqnarray}
where $\theta_{\ell} \sim \mathcal{U}(0,2\pi)$, $\eta$ is given in \eqref{def:eta}, and $o_p(1)$ term vanishes as $K \ra \infty$. Note that  for a fixed $\ell$,  $(T_{\ell,r})_r$  are independent. Hence, given $\mathbf{X}=\mathbf{x}$,
\begin{IEEEeqnarray}{c}
\sum \limits_{r\neq \ell}(\mathrm{C}_{1})_{\ell r} x_r T_{\ell r} \stackrel{d}{=} \norm{\mathbf{x}}_4 T_{\ell},\end{IEEEeqnarray}
 where $T_{\ell} \sim \mathcal{N}_{\mathbb{C}}(0,\sigma^2_{T_{\ell}})$, and 
\begin{IEEEeqnarray}{rcl}
\sigma^2_{T_{\ell}}&=&\frac{1}{\norm{\mathbf{x}}_4^2}  \sum \limits_{r\neq \ell}|(\mathrm{C}_{1})_{\ell r}|^2 |x_r|^2 \nonumber \\
&\stackrel{(a)}{\leq}&  \sqrt{\sum \limits_{r\neq \ell} |(\mathrm{C}_{1})_{\ell r}|^4} \nonumber \\
&\stackrel{(b)}{\leq}& \rho.
\end{IEEEeqnarray}
Step $(a)$ is derived using the
Cauchy–Schwarz inequality. Step $(b)$ follows from the structure of $\mathrm{C}_1$. Thus, conditioned on $\vc{X}=\vc{x}$,
\begin{IEEEeqnarray}{c}
Y_{K,\ell}\stackrel{d}{=} e^{\zeta+j\theta_{\ell}} x_{\ell}+\frac{e^{j\theta_{\ell}}\norm{\mathbf{x}}_4}{\sqrt{K}}T_{\ell}+\sqrt{\eta}~Z_{\ell}+o_p(1). \end{IEEEeqnarray}

Now, the conditional entropy can be bounded as following 
\begin{IEEEeqnarray}{rcl}
 h(Y_{K,\ell}|\mathbf{X}) &\leq &  h\left(e^{\zeta+j\theta_{\ell}} X_{\ell}+\frac{e^{j\theta_{\ell}}\norm{\mathbf{X}}_4}{\sqrt{K}}T_{\ell}\Big|\mathbf{X}\right)+O\left(\sqrt{\frac{K}{\mathcal{P}}}\right)\nonumber \\
&= & \log_2(\pi)+O\left(\sqrt{\frac{K}{\mathcal{P}}}\right)\nonumber\\
&&  +h\left(\Bigg| e^{\zeta} X_{\ell}+\frac{\norm{\mathbf{X}}_4}{\sqrt{K}}T_{\ell} \Bigg|^2\Big| \mathbf{X}\right) \nonumber\\
&=&  \log_2(\pi)+O\left(\sqrt{\frac{K}{\mathcal{P}}}\right)\nonumber\\
&&+h\bigg(\frac{2e^{\zeta} |X_{\ell}|\norm{\mathbf{X}}_4}{\sqrt{K}}\mathfrak{R}\left(T_{\ell} e^{-j\angle X_{\ell}}\right)+\frac{\norm{\mathbf{X}}_4^2}{K}|T_{\ell}|^2\Big| \mathbf{X}\bigg)\nonumber\\
&=& \log_2(\pi)+O\left(\sqrt{\frac{K}{\mathcal{P}}}\right)+O\left(\frac{1}{\sqrt{K}}\right)\nonumber\\
&&+h\left(\frac{2e^{\zeta} |X_{\ell}|\norm{\mathbf{X}}_4}{\sqrt{K}}\mathfrak{R}\left(T_{\ell} e^{-j\angle X_{\ell}}\right)\Big| \mathbf{X}\right)\nonumber\\
&=&  \log_2\left(\frac{2e^{\zeta}\pi}{\sqrt{K}}\right)+O\left(\sqrt{\frac{K}{\mathcal{P}}}\right)+O\left(\frac{1}{\sqrt{K}}\right)\nonumber\\
&&+\mathbb{E}\left[\log_2\left(|X_{\ell}|\norm{\mathbf{X}}_4\right)\right]+ h\left(\mathfrak{R}(T_{\ell})\big| \mathbf{X}\right)\nonumber\\
&\leq&  \log_2\left(\frac{2e^{\zeta}\pi}{\sqrt{K}}\right)+O\left(\sqrt{\frac{K}{\mathcal{P}}}\right)+O\left(\frac{1}{\sqrt{K}}\right)\nonumber\\
&&+ \mathbb{E}\left[\log_2\left(|X_{\ell}|\norm{\mathbf{X}}_4\right)\right]+\frac{1}{2}\log_2\left(2\pi e \rho  \right).
\end{IEEEeqnarray}
This relation together with \eqref{eq:outputEntropy} and \eqref{eq:cEntropySum}  yields the first part of the theorem.

\paragraph*{Part ii}For $\vc{X}\sim \mathcal{N}_{\mathbb{C}}(0,\mathcal{P}\mathrm{I}_n)$, 
\begin{IEEEeqnarray}{c}
h\left(\mathsf{R}(\boldsymbol{\theta})\mathbf{X}\right)=h\left(\mathbf{X}\right)=n\log_2\left( 2\pi e \mathcal{P} \right). \label{eq:outputEntropyGaussian}
\end{IEEEeqnarray}
 Moreover, 
\begin{IEEEeqnarray}{rcl}
&\mathbb{E}&\left[\log_2\left(|X_1|\right)+\log_2\left(\norm{\mathbf{X}}_4\right)\right]\nonumber\\
&=& \frac{1}{2}\log_2\left(\mathbb{E}\left[|X_1|^2\right]\right)+\frac{1}{4}\log_2\left(\mathbb{E}\left[\norm{\mathbf{X}}_4^4\right]\right)+O\left(\frac{1}{\mathcal{P}}\right)\nonumber \\
&=&\frac{1}{2} \log_2(\mathcal{P})+\frac{1}{4}\log_2\left(2n\mathcal{P}^2\right)+O\left(\frac{1}{\mathcal{P}}\right)\nonumber \\
&=&\log_2(\mathcal{P})+\frac{1}{4}\log_2\left(2n\right)+O\left(\frac{1}{\mathcal{P}}\right).
\end{IEEEeqnarray}
This, together with \eqref{eq:outputEntropyGaussian} and the first part of theorem implies that
\begin{IEEEeqnarray*}{c}
\frac{1}{n} I(\vc{X};\vc{Y}_K) \geq\frac{1}{2} \log_2(K)+\frac{1}{2} \log_2\left(\frac{e^{2\zeta+1}}{ \rho \pi \sqrt{8n}}\right)+o(1),
\end{IEEEeqnarray*}where $o(1)$ term tends to zero with $K\rightarrow \infty $ and $\mathcal{P}/K \rightarrow \infty$. Setting $K = \sqrt[\delta]{\snr}\ra\infty$, 
we obtain
\begin{IEEEeqnarray*}{rcl}
 \frac{1}{2 \delta} \log_2\left(1+\snr\right)+\frac{1}{2} \log_2\left(\frac{e^{2\zeta+1}}{\rho\pi\sqrt{8n} }\right)
 &\leq& \frac{1}{n} \lim_{K\ra\infty}I(\vc{X};\vc{Y}_K) 
\\
&= & \frac{1}{n} I(\vc{X}; \lim_{K\ra\infty}\vc{Y}_K) 
\end{IEEEeqnarray*}
Since the left hand side does not depend on $p_X(x)$, we obtain \eqref{eq:fadingLower}. The last equality above follows from the continuity of $I(\vc X, \vc Y_K) $ for the  SSFM channel as a function of $K$ at  $K\ra\infty$, shown in Lemma~\ref{lemm:continuity}.


\subsection{Proof of Lemma~\ref{lem:haarMeasure}} \label{pr:haarMeasure}

The proof is based on Theorem~\ref{th:kawadaItoStromberg} in Appendix~\ref{sec:math}.  The reader is referred to Appendix~\ref{sec:math} for the notation used in this section. 

Let $\mathsf{T}\triangleq\mathrm{D}\mathsf{R}(\boldsymbol{\theta})$. Clearly, $\mathsf{T}$ is a unitary matrix. Denote the probability measure of $ \mathsf{T}$ by $\mu$.  
We show that the following two conditions of  Theorem~\ref{th:kawadaItoStromberg} hold. 

\begin{itemize}[leftmargin=*,wide]
    \item[] \emph{Condition 1.} Denote the smallest closed subgroup of $\mathbb{U}_n$ that contains support of $\mu$, \textit{i.e.} $\mathcal{S}(\mu)$, as $\mathbb{H}$. Moreover, denote the smallest subgroup of block diagonal matrices that contains $\mathrm{D}$ as $\mathbb{U}_n(\nu)$. 
    
    The first condition to verify is $\mathbb{H}=\mathbb{U}_n(\nu)$. 
    This condition is needed, because if $\mathbb{H} \subset \mathbb{U}_n(\nu)$, then the product of instances of $\mathsf{T}$ will not be in $\mathbb{H}$. Hence, the probability measure of the product of $K$ i.i.d. instances of $\mathsf{T}$ would not be a Haar measure on $\mathbb{U}_n(\nu)$.

    By letting $\boldsymbol{\theta}=(0,\ldots,0)$, we have $\mathrm{D} \in \mathbb{H}$, and thus $\mathrm{D}^{-1} \in \mathbb{H}$, and $\mathrm{D}^{-1}  \mathrm{D}  \mathsf{R}\boldsymbol(\theta)=\mathsf{R}\boldsymbol(\theta)\in \mathbb{H} $. Hence, $\mathbb{H}$ contains all diagonal unitary matrices including the matrix $\mathrm{D}$. Since the smallest block diagonal subgroup that contains $D$ is $\mathbb{U}_n(\nu)$, Theorem~\ref{th:Borevich} in Appendix~\ref{sec:math} implies $\mathbb{H}=\mathbb{U}_n(\nu)$.

    \item[]\emph{Condition 2.} The next condition to verify is that $\mu$ is not normally aperiodic. This means that $\mathcal{S}(\mu)$ is not contained in a (left or right) coset of a proper closed normal subgroup of $\mathbb{U}_n(\nu)$. To see why this condition is needed, by contradiction suppose that there exists a proper closed normal subgroup $\mathbb{H}$ of $\mathbb{U}_n(\nu)$ and $\mathrm{V} \in \mathbb{U}_n(\nu)$, such that $\mathcal{S}(\mu) \subseteq \mathrm{V}\mathbb{H}$ or $\mathcal{S}(\mu) \subseteq \mathbb{H}\mathrm{V}$, or equivalently $\mathrm{V}^{-1}\mathcal{S}(\mu) \subseteq \mathbb{H}$ or $\mathcal{S}(\mu) \mathrm{V}^{-1} \subseteq \mathbb{H}$. Suppose that  
    \begin{IEEEeqnarray}{c}
    V^{-1}=
        \mathrm{D} \mathrm{R}(\theta_r)  \cdots   \mathrm{D} \mathrm{R}(\theta_1).  \label{eq:firstRsteps}
    \end{IEEEeqnarray}
    Consider all matrices 
    $\mathsf M=\prod\limits_{i=r+1}^{\infty}\mathsf{D} \mathrm{R}(\boldsymbol{\theta}_j) V^{-1}$.
    The second condition states that the smallest closed normal subgroup that contains these matrices is $\mathbb{U}_n(\nu)$. In other words,  starting from any $r$ initial steps, all possible unitary matrices in $\mathbb{U}_n(\nu)$ can be reached.

    To verify the second condition, we consider the cases $\mathrm{V}^{-1}\mathcal{S}(\mu) \subseteq \mathbb{H}$ or $\mathcal{S}(\mu)\mathrm{V}^{-1} \subseteq \mathbb{H}$ separately.
    
    Left Coset: In this case, $\mathrm{V}^{-1}\mathrm{D}$ and the subgroup of diagonal matrices belong to $\mathbb{H}$. Suppose that there exists a $\mathrm{W} \in \mathbb{U}_n(\nu)$ such that $\mathrm{W} \notin \mathbb{H}$ and $\mathrm{W}=\mathrm{Q}\Gamma \mathrm{Q}^{-1}$, where $\Gamma$ is a diagonal matrix. However, since $\mathbb{H}$ is a normal subgroup of $\mathbb{U}_n(\nu)$ and $\Gamma \in \mathbb H$, then $\mathrm{W}=\mathrm{Q}\Gamma \mathrm{Q}^{-1} \in \mathbb{H}$, which is a contradiction.

    Right Coset: Since $\mathrm{D}\mathsf{R}(\boldsymbol{\theta}_1)\mathrm{V}^{-1} \in \mathbb{H}$ and $\mathrm{D}\mathsf{R}(\boldsymbol{\theta}_2)\mathrm{V}^{-1}\in \mathbb{H}$, we have $\mathrm{D}\mathsf{R}(\boldsymbol{\theta}_1-\boldsymbol{\theta}_2)\mathrm{D}^{-1} \in \mathbb{H}$. Hence, for any  $\boldsymbol{\theta}$,  $\mathrm{D}\mathsf{R}(\boldsymbol{\theta})\mathrm{D}^{-1} \in \mathbb{H}$. Similar to the previous case, by contradiction suppose that there exists $\mathrm{W} \in \mathbb{U}_n(\nu)$ such that $\mathrm{W} \notin \mathbb{H}$ and $\mathrm{W}=\mathrm{Q}\Gamma \mathrm{Q}^{-1}$, where $\Gamma$ is a diagonal matrix. However, since $\mathbb{H}$ is a normal subgroup of $\mathbb{U}_n(\nu)$, $\mathrm{Q}\mathrm{D}^{-1} \in \mathbb{U}_n(\nu)$, and $ \mathrm{D}\Gamma \mathrm{D}^{-1} \in \mathbb{H}$, thus $\mathrm{W}=(\mathrm{Q}\mathrm{D}^{-1}) \mathrm{D}\Gamma \mathrm{D}^{-1} (\mathrm{D}\mathrm{Q}^{-1}) \in \mathbb{H}$, which is a contradiction.
\end{itemize}


\subsection{Proof of Lemma~\ref{lem:mainThSSFM-ex}} \label{pr:mainThSSFM-ex}

The proof of Lemma~\ref{lem:mainThSSFM-ex} is similar to the proof of Lemma~\ref{lem:asymFadingK}, where $M_K$ is expanded in $1/K$. 

Note that, if $D_K$ does not depend on $K$, when $\mathcal P\ra\infty$, phase tends to a uniform random variable in every segment for every input \cite{yousefi2016cap}. However, if dispersion values scale as $1/K$ and $\mathcal{P}=(\sigma^2 \mathcal{L})K^{\delta}$, phase tends to zero in one segment if $\delta<1$. But, if we add sufficiently large number $K^{1-\delta+0^+}$ of segments, so that the variance of phase tends to infinity, output phase tends to a uniform variable for every input. In what follows, we make these statements precise.

We prove the lemma formally by induction on the segment index $i$. The output $\vc{V}_{i+1}$ of the segment $i\in [K]$ as a function of the channel input $X$ is 
\begin{IEEEeqnarray}{c}
    \mathbf{V}_{i+1}=\mathsf{M}_{i,K}\vc{X}+\vc{Z}_{i}, \label{eq:ViX}
\end{IEEEeqnarray}
where
\begin{IEEEeqnarray}{rCl}
    \mathsf{M}_{i,K}
    &\triangleq& \prod \limits_{t=1}^{i} \left( \mathrm{D}_K \mathsf{\Phi}(t) \right)
    \nonumber\\
    &=&\mathrm{D}_K \mathsf{\Phi}(i) \mathsf{M}_{i-1,K},
\end{IEEEeqnarray}
in which
\begin{IEEEeqnarray}{c}
\mathsf{\Phi}(i)\triangleq \diag\Bigl( \bigl(\exp(j\Phi_{i,\ell}) \bigr)_{\ell=1}^n \Bigr),  \end{IEEEeqnarray}
where the nonlinear phase $\Phi_{i,\ell}$ is given in \eqref{eq:Phi-ij}. Further,
\begin{IEEEeqnarray*}{rCl}
    \vc{Z}_{i} &\triangleq&  \sum \limits_{t=1}^i \left( \left(\prod \limits_{s=t}^{i} \left( \mathrm{D}_K  \mathsf{\Phi}(s)\right)\right) \vc{\bar{Z}}_t\right)
    \\ &=& \mathrm{D}_K  \mathsf{\Phi}(i)\left(\vc{\vc{Z}_{i-1}+\bar{Z}}_i\right), 
\end{IEEEeqnarray*}
where $ \mathsf{M}_{0,K} \triangleq \mathrm{I}_n$ and $\vc{Z}_0\triangleq 0$. Note that $\mathsf{M}_{K,K}=\mathsf{M}_K$.

First, we expand $\mathsf{M}_{i,K}$ similar to the analysis in the proof of Lemma~\ref{lem:asymDistK}. For $t \in [i]$, denote 
\begin{IEEEeqnarray*}{c}
\mathsf{R}_i(t)\triangleq \diag\left( \Bigl( \exp\bigl( j \sum \limits_{s=t}^i \Phi_{s,\ell} \bigr)\Bigr)_{\ell=1}^n \right),
\end{IEEEeqnarray*}
\begin{IEEEeqnarray}{c}
\mathsf{L}_i(t)\triangleq \mathsf{R}_i(t) \mathrm{C}_1 \mathsf{R}_i(t)^{-1},\end{IEEEeqnarray}
and
\begin{IEEEeqnarray}{c}
\Delta_i \triangleq \frac{1}{K}\sum \limits_{t=1}^{i}\mathsf{L}_i(t)-\frac{i}{K}(\zeta+jd)\mathrm{I}_n,
\end{IEEEeqnarray}
where $\mathrm{C}_1$ is defined in \eqref{eq:C-i}. Expand $\mathsf{M}_{i,K}$ as:
\begin{IEEEeqnarray}{rcl}
\mathsf{M}_{i,K}
&=&\left\{\prod \limits_{t=2}^{i} \left( \mathrm{D}_K   \left[\mathsf{R}_i(t) ,\mathrm{D}_K\right]\right)\right\}  \mathrm{D}_K \mathsf{R}_i(1)\nonumber\\
&=& \left\{\prod \limits_{t=2}^{i} \left( \mathrm{I}_n+\frac{1}{K}\mathsf{L}_i(t)\right)\right\} \mathsf{R}_i\left(1\right)+O_p\left(K^{-1}\right)\nonumber \\
  &=& \left(e^{\frac{i}{K}(\zeta+jd)}\mathrm{I}_n+\frac{1}{K}\sum \limits_{t=1}^{i}\mathsf{L}_i(t)-\frac{i}{K}(\zeta+jd)\mathrm{I}_n\right)  \mathsf{R}_i\left(1\right)\nonumber \\
  &&+O_p\left(K^{-1}\right)\nonumber\\
   &=& e^{\frac{i}{K}(\zeta+jd)}\mathsf{R}_i(1)+\Delta_i \mathsf{R}_i(1)+O_p\left(K^{-1}\right).
  \label{eq:Mr-K}
\end{IEEEeqnarray}
Note that $\mathsf{R}_K(1)=\mathsf{S}_K$.

Fix $\epsilon'>0$ sufficiently small. We shall prove that for each $i$:

$\bullet$ \emph{Claim 1.} For $\ell,\ell' \in [n]$,
\begin{IEEEeqnarray}{c}
    \left(\Delta_i\right)_{\ell,\ell'}=O_p\left(K^{-\underline{\upsilon}(\delta)}\right), \label{eq:inductionAssumption}
\end{IEEEeqnarray}
where $\underline{\upsilon}(\delta)\geq \upsilon(\delta)-\epsilon'$ and $\upsilon(\delta)$ is defined in \eqref{def:upsilonDelta}.

$\bullet$ \emph{Claim 2.} 
For $\ell \in [n]$ and $t \in [i]$,
\begin{IEEEeqnarray}{c}
    \sum \limits_{s=t}^i \Phi_{s,\ell}-(i-t+1)\gamma \varepsilon |X_{\ell}|^2 \stackrel{p}{\ra} 0, \label{eq:arithSum}
\end{IEEEeqnarray}
 when 
\begin{IEEEeqnarray}{c}
    i \leq \begin{cases} 
    \IEEEstrut
    K^{1-\frac{\delta}{3}-\epsilon'},& 0 \leq \delta \leq 1.5,\\
    K^{2-\delta-\epsilon'},& 1.5 \leq \delta \leq 2.
    \IEEEstrut
    \end{cases}\label{eq:iDelta}
\end{IEEEeqnarray}

$\bullet$ \emph{Claim 3.} 
If $i$ satisfies \eqref{eq:iDelta}, then $\underline{\upsilon}(\delta)$ in \eqref{eq:inductionAssumption} is bounded by
    \begin{IEEEeqnarray}{c}
   \underline{\upsilon}(\delta)\geq \begin{cases} 
    \IEEEstrut
    \delta,& 0 \leq \delta \leq 1,\\
    1-g,& 1 \leq \delta \leq 2,
    \IEEEstrut
    \end{cases}\label{eq:uDelta}
\end{IEEEeqnarray}
where 
 \begin{IEEEeqnarray}{c}
g \triangleq \max \limits_{l,l' \in [n]} -\log_K\left(e^{j(|X_l|^2-|X_{l'}|^2)/K}-1\right). \label{def:g}
\end{IEEEeqnarray}
Note that $g=o_p(1)$, and thus vanishes for absolutely continuous inputs, when $\delta>1$.

Lemma~\ref{lem:mainThSSFM-ex} follows  from \eqref{eq:Mr-K} at $i=k$ and  Claim 1. Claim 2 and 3 are needed in the proof of Claim 1.

Note that above Claim 1--3 yield
\begin{IEEEeqnarray}{c}
    \lim \limits_{K \ra \infty}\vc{Z}_K \stackrel{(d)}{=}
  \vc{Z},
 \end{IEEEeqnarray}
where $\vc{Z}\in \mathbb{C}^n$ and $\vc{Z}\sim \mathcal{N}_{\mathbb{C}}\left(0,\eta \sigma^2\mathcal{L}\mathrm{I}_n\right)$.

For $i=1$, Claim 1-3 hold, since
\begin{IEEEeqnarray}{c}
    \mathsf{M}_{1,K}= \mathrm{D}_K \mathsf{\Phi}(1)=e^{\frac{1}{K}(\zeta+jd)}\mathsf{\Phi}(1)+O_p\left(K^{-1}\right).
\end{IEEEeqnarray}
Assume that Claim 1-3 hold for $i\in [r-1]$. We need to show that they hold for $i=r$ as well. 

Denote 
\begin{IEEEeqnarray}{c}
    \vc{E}_t\triangleq e^{-(\zeta+jd)\frac{ t-1}{K}} \mathsf{R}_{t-1}^*(1) \vc{\bar{Z}}_t=O_p\left(K^{-1/2}\right). \label{def:Et}
\end{IEEEeqnarray}
Using the assumption of the induction \eqref{eq:inductionAssumption} together with \eqref{eq:Mr-K}  and \eqref{eq:ViX}, the nonlinear phase $\Phi_{i+1,\ell}$, where $i\in [r-1]$, can be expanded as:
\begin{IEEEeqnarray}{rcl}
    \Phi_{i+1,\ell} 
    &=&\gamma \varepsilon |V_{i+1,\ell}|^2+2 \gamma \varepsilon \sqrt{\varepsilon} \mathfrak{R}(V_{i+1,\ell}^* Z_{i+1,\ell}^{'})+\gamma \varepsilon^2|Z^{''}_{i+1,\ell}|^2 \nonumber \\
     &=&\gamma \varepsilon |V_{i+1,\ell}|^2+O_p\left(K^{\frac{\delta-3}{2}}\right) \nonumber \\
    &\overset{(a)}{=}& \gamma \varepsilon e^{\frac{2i}{K}\zeta}|X_{\ell}|^2+O_p\left(K^{\frac{\delta-3}{2}}\right)\nonumber\\
    &&+2\gamma \varepsilon {\sum \limits_{\ell'\neq \ell}} \mathfrak{R}\left[\left(e^{\frac{ i}{K}(\zeta+jd)}\mathsf{R}_{i}(1)\vc{X}\right)_{\ell}^* (\Delta_{i,K})_{\ell,\ell'}X_{\ell'} \right]\nonumber\\
    &&+2\gamma \varepsilon \mathfrak{R}\left[\left(e^{\frac{ i}{K}(\zeta+jd)}\mathsf{R}_{i}(1)\vc{X}\right)_{\ell}^* Z_{i,\ell} \right] \nonumber \\
    &=& \gamma \varepsilon e^{\frac{2i}{K}\zeta}|X_{\ell}|^2+O_p\left(K^{\frac{\delta-3}{2}}\right)\nonumber\\
    &&{+}2\gamma \varepsilon {\sum \limits_{\ell'\neq \ell}} {\mathfrak{R}}\left[\left(e^{\frac{ i}{K}(\zeta+jd)}\mathsf{R}_{i}(1)\vc{X}\right)_{\ell}^* (\Delta_{i,K})_{\ell,\ell'}X_{\ell'} \right] \nonumber \\
    &&+2\gamma \varepsilon e^{\frac{2i}{K}\zeta} \mathfrak{R}\left[X_{\ell}^* \sum \limits_{t=1}^i  E_{t,\ell} \right]. \label{eq:phiExpansion}
\end{IEEEeqnarray}
Here step $(a)$ follows by substituting \eqref{eq:ViX} and   \eqref{eq:Mr-K} into the previous line.
Variables $Z'_{i+1,\ell}$ and $Z^{''}_{i+1,\ell}$ denote Gaussian noises with variances that do not depend on $K$ and $E_{t,\ell}$ is defined in \eqref{def:Et}.

The term
\begin{IEEEeqnarray}{c}
    2\gamma \varepsilon \sum \limits_{\ell'\neq \ell} \mathfrak{R}\left[\left(e^{\frac{ i}{K}(\zeta+jd)}\mathsf{R}_{i}(1)\vc{X}\right)_{\ell}^* (\Delta_{i,K})_{\ell,\ell'}X_{\ell'} \right],
\end{IEEEeqnarray}
captures intra-channel interactions, while the terms
\begin{IEEEeqnarray}{c}
    2\gamma \varepsilon e^{\frac{2i}{K}\zeta} \mathfrak{R}\left[X_{\ell}^* \sum \limits_{t=1}^i  E_{t,\ell} \right],
\end{IEEEeqnarray} 
and
\begin{IEEEeqnarray}{c}
    2 \gamma \varepsilon \sqrt{\varepsilon} \mathfrak{R}\left(\left(e^{\frac{ i}{K}(\zeta+jd)}\mathsf{R}_{i}(1)\vc{X}\right)_{\ell}^* Z_{i+1,\ell}^{'}\right),
\end{IEEEeqnarray}
represent the signal-noise interactions.

To show Claims~1 and 3, note that $\left(\Delta_r\right)_{\ell,\ell}$, $\ell \in [n]$, is equal to zero, with probability one. It remains to show that the off-diagonal elements are $O_p\left(K^{-\underline{\upsilon}(\delta)}\right)$. We show this for element $(1,2)$; the proof is similar for other elements. This element is equal to
\begin{IEEEeqnarray}{c}
    (\Delta_r)_{1,2}=\frac{(\mathrm{C}_1)_{1,2}}{K}\sum \limits_{i=1}^{r}e^{j \sum \limits_{s=i}^r \left(\Phi_{s,1}-\Phi_{s,2}\right)}. \label{eq:phiGeneral}
\end{IEEEeqnarray}

The rest of the proof is presented for different ranges of $\delta$, and for $\zeta=0$, separately. Since in general, for $i \in [K]$,
\begin{IEEEeqnarray}{c}
    e^{\frac{i}{K}\zeta} \vc{X}=K^{\xi} \vc{X}, \label{eq:lossEffect}
\end{IEEEeqnarray}
where 
\begin{IEEEeqnarray}{c}
    \xi=\frac{2i\zeta}{K\ln(K)} \rightarrow 0.
\end{IEEEeqnarray} 
It can be shown that asymptotically as $K\ra\infty$ the effect of loss in the convergence rate of $\Delta_r$ vanishes.

$\bullet$ \emph{Case $0<\delta<1$:} In this case, first for $r \leq K^{1-\delta/3-\epsilon'}$, we prove Claims 2 and 3, and consequently Claim 1 follows. Then using this result, we prove Claim 1 for $r > K^{1-\delta/3-\epsilon'}$ as well.

Assume that $r \leq K^{1-\delta/3-\epsilon'}$.
We argue first for $t \in [r]$, Claim~2 holds, from which Claims~1 and 3 are then concluded. Let 
\begin{IEEEeqnarray}{c}
Q_{\ell}\triangleq \gamma \mathcal{L} K^{-\delta}|X_{\ell}|^2,~\ell \in [n].
\end{IEEEeqnarray}
Using the induction assumption \eqref{eq:inductionAssumption} and since $r-t+1 \leq K^{1-\delta/3-\epsilon'}$,
\begin{IEEEeqnarray}{rl} 
\sum \limits_{s=t}^r& \Phi_{s,\ell} \nonumber\\
=&(r-t+1)K^{\delta-1}Q_{\ell} \nonumber
 \\
&+2\gamma \varepsilon  \mathfrak{R}\left[X_{\ell}^* \left(\sum \limits_{s=1}^{t-1} (r-t+1) E_{s,\ell} +\sum \limits_{s=t}^{r-1} \left(r-s\right) E_{s,\ell} \right)\right] \nonumber\\
 &+O_p\left((r-t+1)K^{\delta-1-\underline{\upsilon}(\delta)}\right)+O_p\left((r-t+1)K^{\frac{\delta-3}{2}}\right) \nonumber\\
 \stackrel{(a)}{=}&(r-t+1)K^{\delta-1}Q_{\ell}+O_p\left(K^{-\delta/3-2\epsilon'}\right)+O_p\left(K^{\frac{\delta-3}{6}-\epsilon'}\right) \nonumber \\&+O_p\left(K^{\delta/2-1-1/2}A\right)  \nonumber  \\=&(r-t+1)K^{\delta-1}Q_{\ell} +O_p\left(K^{-\delta/3-2\epsilon'}\right)+O_p\left(K^{\frac{\delta-3}{6}-\epsilon'}\right)\nonumber\\
&+O_p\left(K^{-3\epsilon'/2}\right)\nonumber\\
\ra& (r-t+1)K^{\delta-1}Q_{\ell}, \label{eq:secondOrderD1}
\end{IEEEeqnarray}
where $(a)$ holds due to assumption of induction \eqref{eq:uDelta} and
\begin{IEEEeqnarray}{rcl}
    A &\triangleq& \Bigl((t-1)(r-t+1)^2 + \sum_{\ell=1}^{r-t} \ell^2\Bigr)^{1/2}\nonumber \\
    &=&O\left(r^{3/2}\right).
\end{IEEEeqnarray}

This proves Claim 2. Note that as can be seen in the above calculations, due to signal-noise interactions, at most $ K^{1-\delta/3-\epsilon'}$ consecutive terms $\sum \limits_{s=t}^r \Phi_{s,\ell}$ constitute  arithmetic series. For small values of noise, it can be seen that the above argument almost holds for any $r \leq K^{1-\epsilon'}$, except for very large values of $K$,  when $K \gg \left(\gamma \mathcal{L}^{3/2} \sigma\right)^{-2/\delta}$.

Now, to show Claim~3, by denoting $Q\triangleq Q_1-Q_2$ and due to \eqref{eq:secondOrderD1},
\begin{IEEEeqnarray}{rcl}
    (\Delta_r)_{1,2}&=&\frac{(\mathrm{C}_1)_{1,2}}{K}\sum \limits_{i=1}^{r}e^{j \sum \limits_{s=i}^r \left(\Phi_{s,1}-\Phi_{s,2}\right)}\nonumber \\
    &=& \frac{(\mathrm{C}_1)_{1,2}e^{jK^{\delta-1}Q}\left(e^{jK^{2\delta/3-\epsilon'}Q}-1\right)}{K\left(e^{jK^{\delta-1}Q}-1\right)}.
\end{IEEEeqnarray}
Furthermore, since for $0 \leq \delta \leq 1$,
\begin{IEEEeqnarray}{c}
e^{jK^{\delta-1}Q}-1=O_p\left(K^{\delta-1}\right),  
\end{IEEEeqnarray}
then
\begin{IEEEeqnarray}{c}
    \big|(\Delta_r)_{1,2}\big| = \big|(\mathrm{C}_1)_{1,2}\big| ~O_p\left(K^{-\delta}\right),
\end{IEEEeqnarray}
which completes the proof of Claim 3 and consequently Claim 1.  

For $r > K^{1-\delta/3-\epsilon'}$, similarly it can be verified that for each $K^{1-\delta/3-\epsilon'}$ consecutive terms, 
\begin{IEEEeqnarray}{c}
\sum \limits_{i=t_1}^{t_2}e^{j \sum \limits_{s=i}^{r} \left(\Phi_{s,1}-\Phi_{s,2}\right)}=O_p\left(K^{1-\delta}\right),
\end{IEEEeqnarray}
where $t_2=t_1+K^{1-\delta/3-\epsilon'}$.
Hence
\begin{IEEEeqnarray}{rcl}
 \big|(\Delta_r)_{1,2}\big|&=&\big|(\mathrm{C}_1)_{1,2}\big| ~O_p\left(K^{\delta/3+\epsilon'}K^{1-\delta}K^{-1}\right)\nonumber\\
 &=&\big|(\mathrm{C}_1)_{1,2}\big| ~O_p\left(K^{-2\delta/3+\epsilon'}\right).
\end{IEEEeqnarray}
The above bound can be improved using the following approach. When $t$ consecutive terms form a geometric series, then
\begin{IEEEeqnarray}{rl}
    \sum \limits_{i=m-t+1}^{m}&e^{j \sum \limits_{s=i}^{r} \left(\Phi_{s,1}-\Phi_{s,2}\right)}\nonumber \\
    &= \frac{e^{j \sum \limits_{s=m-t }^{r} \left(\Phi_{s,1}-\Phi_{s,2}\right)}-e^{j \sum \limits_{s=m}^{r} \left(\Phi_{s,1}-\Phi_{s,2}\right)}}{e^{jK^{\delta-1}Q}-1}.
\end{IEEEeqnarray}

On the other hand, if $t_2-t_1 = K^{ 1-\delta/3+\epsilon'}$, then similar to the analysis in \eqref{eq:secondOrderD1},
\begin{IEEEeqnarray}{rl} 
\sum \limits_{s=t_1}^{t_2-1}& \Phi_{s,\ell} \nonumber\\
=&(t_2-t_1)K^{\delta-1}Q_{\ell} \nonumber
 \\
&+2\gamma \varepsilon  \mathfrak{R}\left[X_{\ell}^* \left(\sum \limits_{s=1}^{t_1-1} (t_2-t_1) E_{s,\ell} +\sum \limits_{s=t_1}^{t_2-1} \left(t_2-s\right) E_{s,\ell} \right)\right] \nonumber\\
 &+O_p\left((t_2-t_1)K^{\delta-1-\underline{\upsilon}(\delta)}\right)+O_p\left((t_2-t_1)K^{\frac{\delta-3}{2}}\right). 
 \end{IEEEeqnarray}
The second term on the RHS of the above relation is non-deterministic given the input and is of order 
\begin{IEEEeqnarray}{c}
\Omega_p \left((t_2-t_1)^{\frac{3}{2}} K^{\frac{\delta-3}{2}}\right)  = \Omega_p\left( K^{\frac{3\epsilon'}{2}}\right). 
\end{IEEEeqnarray} 
This concludes that this term grows as $K \ra \infty$.  It can be verified that in general if $A=B\cdot Z$, where $B$ is a constant and $Z$ is a  random variable with continuous PDF, then as $B \ra \infty$, $\mod(A,2\pi) \ra \mathcal{U}(0,2\pi)$. Consequently
\begin{IEEEeqnarray}{c}
    \mod\left(\sum \limits_{s=t_1}^{t_2-1} \left(\Phi_{s,1}-\Phi_{s,2}\right),2\pi \right) \ra \mathcal{U}(0,2\pi).
\end{IEEEeqnarray}
Hence, the summation of terms with a distance more than $K^{ 1-\delta/3+\epsilon'}$ can be considered as the summation of independent random variables. Using central limit theorem  yields
\begin{IEEEeqnarray}{rcl}
 \big|(\Delta_r)_{1,2}\big|&=& \big|(\mathrm{C}_1)_{1,2}\big| ~O_p\left(K^{2\epsilon'}K^{\delta/6}K^{1-\delta}K^{-1}\right)\nonumber \\
 &=&O_p\left(K^{-5\delta/6+2\epsilon'}\right).
\end{IEEEeqnarray}
This completes the proof of Claim 1.

$\bullet$ \emph{Case $1\leq \delta<1.5$:} 
Similar to the previous case, first assume that $r\leq K^{1-\delta/3-\epsilon'}$. For $i \in [r-1]$, due to the assumption of the induction \eqref{eq:uDelta},  $\underline{\upsilon}(\delta)\geq 1-g$. Since $1-\delta/3<2-\delta$, then it can be verified similar to the previous case that the restrictive term is the signal-noise interaction term and \eqref{eq:arithSum} (Claim 2) holds. Hence, 
\begin{IEEEeqnarray}{c}
    \sum \limits_{i=1}^{r}e^{j \sum \limits_{s=i}^r \left(\Phi_{s,1}-\Phi_{s,2}\right)}
    = \frac{e^{j\left(\Phi_{1,1}-\Phi_{2,2}\right)}\left(e^{jr K^{\delta-1}Q}-1\right)}{K\left(e^{jK^{\delta-1}Q}-1\right)}.~\label{eq:oscillation}
\end{IEEEeqnarray}
For $\delta>1$, the term $1/(e^{jK^{\delta-1}Q}-1)$ is an oscillating function which is of order $O_p(K^g)$, which concludes that $\underline{\upsilon}(\delta)=1-g$ for $r$ steps as well, $r\leq K^{1-\delta/3-\epsilon'}$. This proves Claim~3 and consequently Claim~1 for $r\leq K^{1-\delta/3-\epsilon'}$.

For $r> K^{1-\delta/3-\epsilon'}$, similar to the previous case it can be verified that sum of each $K^{1-\delta/3-\epsilon'}$ consecutive terms is $O_p(K^{g})$. Thus,
\begin{IEEEeqnarray}{c}
 \big|(\Delta_r)_{1,2}\big|= \big|(\mathrm{C}_1)_{1,2}\big| ~O_p\left(K^{\delta/3-1+g+\epsilon'}\right).
 \end{IEEEeqnarray}
Furthermore, using the central limit theorem, this bound can be improved to 
\begin{IEEEeqnarray}{c}
 \big|(\Delta_r)_{1,2}\big|= \big|(\mathrm{C}_1)_{1,2}\big| ~O_p\left(K^{\delta/6-1+g+2\epsilon'}\right).
 \end{IEEEeqnarray}
Finally, note that for absolutely continuous inputs, $g=o_p(1)$ and  $O_p\left(K^{\delta/6-1+g+2\epsilon'}\right)=O_p\left(K^{\delta/6-1+2\epsilon'}\right)$. This completes the proof of Claim 1.

$\bullet$ \emph{Case $1.5\leq \delta<2$:} 
Since $2-\delta<1-\delta/3$, in this regime, the intra-channel term is the restrictive term. In this case, it can be verified that  $r\leq K^{2-\delta-\epsilon'}$ consecutive terms form geometric series and their sum is $O_p(K^g)$ (Claims 2 and 3). 
Hence, to show Claim 1, $\big|(\Delta_r)_{1,2}\big|$ can be bounded as:
\begin{IEEEeqnarray}{rcl}
 \big|(\Delta_r)_{1,2}\big|&=& \big|(\mathrm{C}_1)_{1,2}\big| ~O_p\left(K^{2\epsilon'+g}K^{\frac{1-(2-\delta)}{2}}K^{-1}\right)\nonumber \\
 &=&\big|(\mathrm{C}_1)_{1,2}\big| ~O_p\left(K^{\delta/2-1.5+g+2\epsilon'}\right).
 \end{IEEEeqnarray}
 Note that similar to the previous case, $g=o_p(1)$ for absolutely continuous inputs and hence $O_p\left(K^{\delta/2-1.5+g+2\epsilon'}\right)=O_p\left(K^{\delta/2-1.5+2\epsilon'}\right)$.

$\bullet$ \emph{Cases $2\leq \delta< 3$ and $3\leq\delta$:}   For these cases, first we argue that when \eqref{eq:inductionAssumption} holds, then as $K \ra \infty$, $\mod\left(\Phi_{r,\ell},2\pi\right)\ra \mathcal{U}(0,2\pi)$, independent of $\left(\Phi_{i,\ell}\right)_{i=1}^{r-1}$. In other words, in this regime as $K \ra \infty$, SSFM channel tends to the finite dispersion fading channel, except for the first segment and when $2<\delta <3$, which is negligible. This, concludes that 
\begin{IEEEeqnarray}{c}
    \Delta_r=\frac{O_p\left(\sqrt{r}\right)}{K}=O_p\left(K^{-1/2}\right),
\end{IEEEeqnarray}
which completes the proof of Claim 1. Note that Claims 2 and 3 are valid only for $\delta<2$.

For $2\leq \delta<3$, the second phase operator is equal to
\begin{IEEEeqnarray}{rcl}
    \Phi_{2,\ell} &=& \gamma \mu\sum \limits_{t=1}^{M} \Big|X_{\ell}+\frac{\left(\mathrm{C}_1\vc{X}\right)_{\ell}}{K}+\vc{Z}_1+O_p\left(\frac{1}{K^2}\right) +W_{2,\ell}(t)\Big|^2 \nonumber \\
    &\approx& \gamma \varepsilon |X_{1,\ell}|^2+\frac{\gamma \mathcal{L}}{K^2}X_{1,\ell}\left(\mathrm{C}_1\vc{X}\right)_{\ell} .
\end{IEEEeqnarray}
For i.i.d. inputs, the second term induces a stochastic impairment that grows if $\delta > 2$ and when $K \ra \infty$.

Hence, as $K \ra \infty$, each step will be reduced to uniform phase noise. Similarly, after $r$ steps, there would be a stochastic impairment of order $\frac{\sqrt{i}}{K}\|X\|_4$. Thus, the channel (except for the first segment) is equivalent to the fading channel.

For $3\leq \delta$,  SSFM channel tends to the fading channel for any input distribution that escapes to infinity, as $K \ra \infty$. To show this, first using the assumption of the induction, for $1 \leq i \leq r$, we have 
\begin{IEEEeqnarray}{c}
|V_{i,\ell}|=K^{\delta/2} |\vc{X}'_\ell|+O_p\left(K^{\delta/2-0.5}\right).
\end{IEEEeqnarray}
Hence, the second term of $ \Phi_{i,1}$, \textit{i.e.} $2 \gamma \varepsilon \sqrt{\varepsilon} \mathfrak{R}(V_{i,\ell} Z_{i,\ell}^{'*})$, which conditioned on other segments and coordinates is not deterministic, grows unboundedly. Thus, $\mod\left(\Phi_{i,\ell},2\pi\right) \ra \mathcal{U}(0,2\pi)$. This completes the proof.

\subsection{Proof of Theorem~\ref{th:mainThSSFM}}\label{pr:mainThComp}

First we show \eqref{eq:capacity-a} holds. Let $\vc{Y}'_K\triangleq \mathsf{R}(\boldsymbol{\theta}')\vc{Y}_K$, where $\theta'_{\ell} \widesim{\text{i.i.d}}\mathcal{U}(0,2\pi)$. Due to data processing inequality
\begin{IEEEeqnarray}{c}
       I(\vc{Y}_K;\vc{X}) \geq I(\vc{Y}'_K;\vc{X}).
\end{IEEEeqnarray}
In the following, we establish a lower bound on the channel $\vc{X} \mapsto \vc{Y}'_K$.

The channel $\vc{X} \mapsto \vc{Y}'_K$ when $K \ra \infty$, is 
\begin{IEEEeqnarray}{c}
\vc{Y}'_K=\mathsf{M}'_K \vc{X}+\vc{Z}, \label{eq:Mprime}
\end{IEEEeqnarray}
where 
\begin{IEEEeqnarray}{c}
    \mathsf{M}'_K\stackrel{p}{=}e^{\zeta}\mathsf{R}(\boldsymbol{\theta})+O_p\left(K^{-\upsilon(\delta)+\epsilon'}\right),
\end{IEEEeqnarray}
where $\theta_l \sim \mathcal{U}(0,2\pi)$, independent of $\vc{X}$.

If $\upsilon(\delta)<\frac{1}{2}\delta$, then for sufficiently small value of $\epsilon'$, the term $O\left(K^{-\upsilon(\delta)+\epsilon'}\right)\vc{X}$ vanishes as $K \ra \infty$ and the channel tends to $n$ independent phase noise channels. The lower bound \eqref{eq:capacity-a} on $\mathcal{C}(\snr,\sqrt[\delta]{\snr})$ for $0 \leq \delta \leq 3/2$ can be then established similar to Theorem~\ref{th:asymCapFadingFinDis}.

If $\upsilon(\delta)> \frac{1}{2}\delta$, an approach similar to that in the proof of Theorem~\ref{th:fadingHighPowerdG1} can be applied. The output entropy $h(\mathbf{Y}'_K)$ can be bounded as in  \eqref{eq:outputEntropy}. Bounding the conditional part also follows similarly as in the proof of Theorem~\ref{th:fadingHighPowerdG1}, with the difference that the defined variable $T_{\ell}$ is not anymore Gaussian and is a random variable with bounded variance $\sigma^2_{T_{\ell}}$. Applying the maximum entropy theorem and letting $\epsilon' \ra 0$, the conditional entropy can be similarly bounded as
\begin{IEEEeqnarray}{rcl}
 h(Y'_{K,\ell}|\mathbf{X}) &\leq & 
  \log_2\left(\frac{2e^{\zeta}\pi}{K^{\upsilon(\delta)}}\right)+O\left(\sqrt{\frac{K}{\mathcal{P}}}\right)+O\left(\frac{1}{\sqrt{K}}\right)\nonumber\\
&&+ \mathbb{E}\left[\log_2\left(|X_{\ell}|\norm{\mathbf{X}}_4\right)\right]+\frac{1}{2}\log_2\left(2\pi e \sigma^2_{T_{\ell}}  \right).
\end{IEEEeqnarray}
This relation together with  \eqref{eq:outputEntropy} and \eqref{eq:cEntropySum}, and letting $\vc{X}=\mathcal{N}_{\mathbb{C}}(0,\mathcal{P}\mathrm{I}_n)$, shows that \eqref{eq:capacity-a} holds also for $\upsilon(\delta) \geq \frac{1}{2}\delta$.

Now, to show \eqref{eq:bound-a}, first note that similar to the proof of Lemma~\ref{lem:mainThSSFM-ex} for the case of $\delta < 1$,  for any $\epsilon'$ when $\snr/K \ra 0$ and $\snr \ra \infty$, we have
\begin{IEEEeqnarray}{rcl}
    \mathsf{M}_K&=&e^{\zeta+jd} \mathsf{S}_K+O_p\left(K^{-\frac{5}{6}\log_{K}(\snr) + \epsilon'\log_{K}(\snr)}\right) \nonumber \\
    &=&e^{\zeta+jd} \mathsf{S}_K+O_p\left(\snr^{-5/6+ \epsilon'}\right).
\end{IEEEeqnarray}
Next, by considering $\vc{X}\mapsto \vc{Y}'_K$, as in \eqref{eq:Mprime}, we have 
\begin{IEEEeqnarray}{c}
   \mathsf{M}'_K\stackrel{p}{=}e^{\zeta}\mathsf{R}(\boldsymbol{\theta})+O_p\left(\snr^{-5/6+ \epsilon'}\right).
\end{IEEEeqnarray}
Hence, by choosing input as $\mathcal{N}(0,\mathcal{P}\mathrm{I}_n)$, we have
\begin{IEEEeqnarray}{cl}
    \lim_{\snr \ra \infty} \lim_{K \ra \infty} \Big[\mathcal{C}(\snr,K)-\frac{1}{2}\log_2\left(1+ \snr\right)&+\frac{1}{2}\log_2(a)\Big] \nonumber \\
    &\geq 0. 
\end{IEEEeqnarray}
This completes the proof.

\begin{figure*}[t]
\begin{subfigure}{\columnwidth}
  \centering
  \includestandalone[scale=0.4]{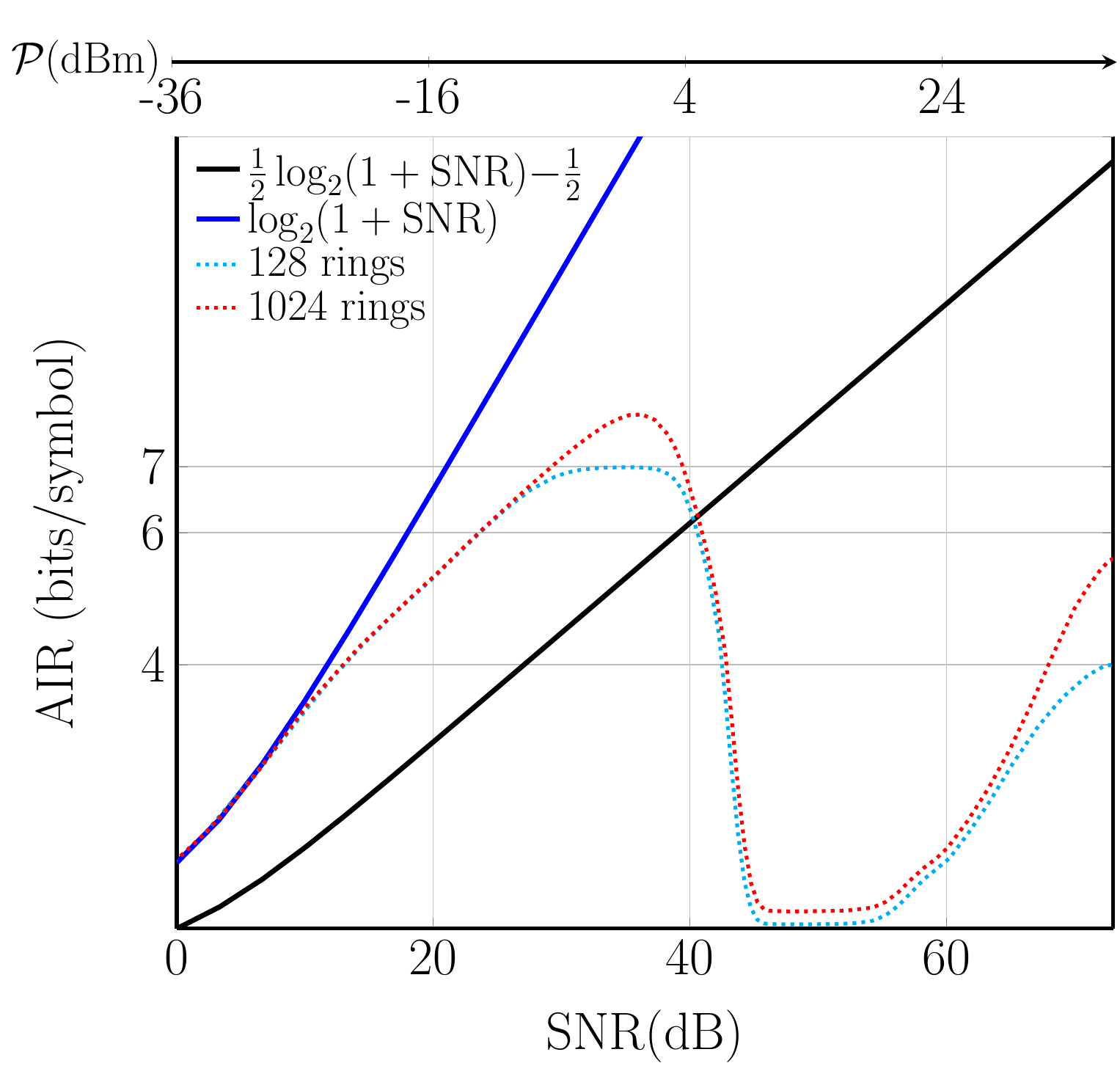}
\caption{$\mathcal{B}=20 \text{ GHz}$} 
 \label{fig:CapacityQAM16B20}
\end{subfigure}
\begin{subfigure}{\columnwidth}
    \centering
\includestandalone[scale=0.4]{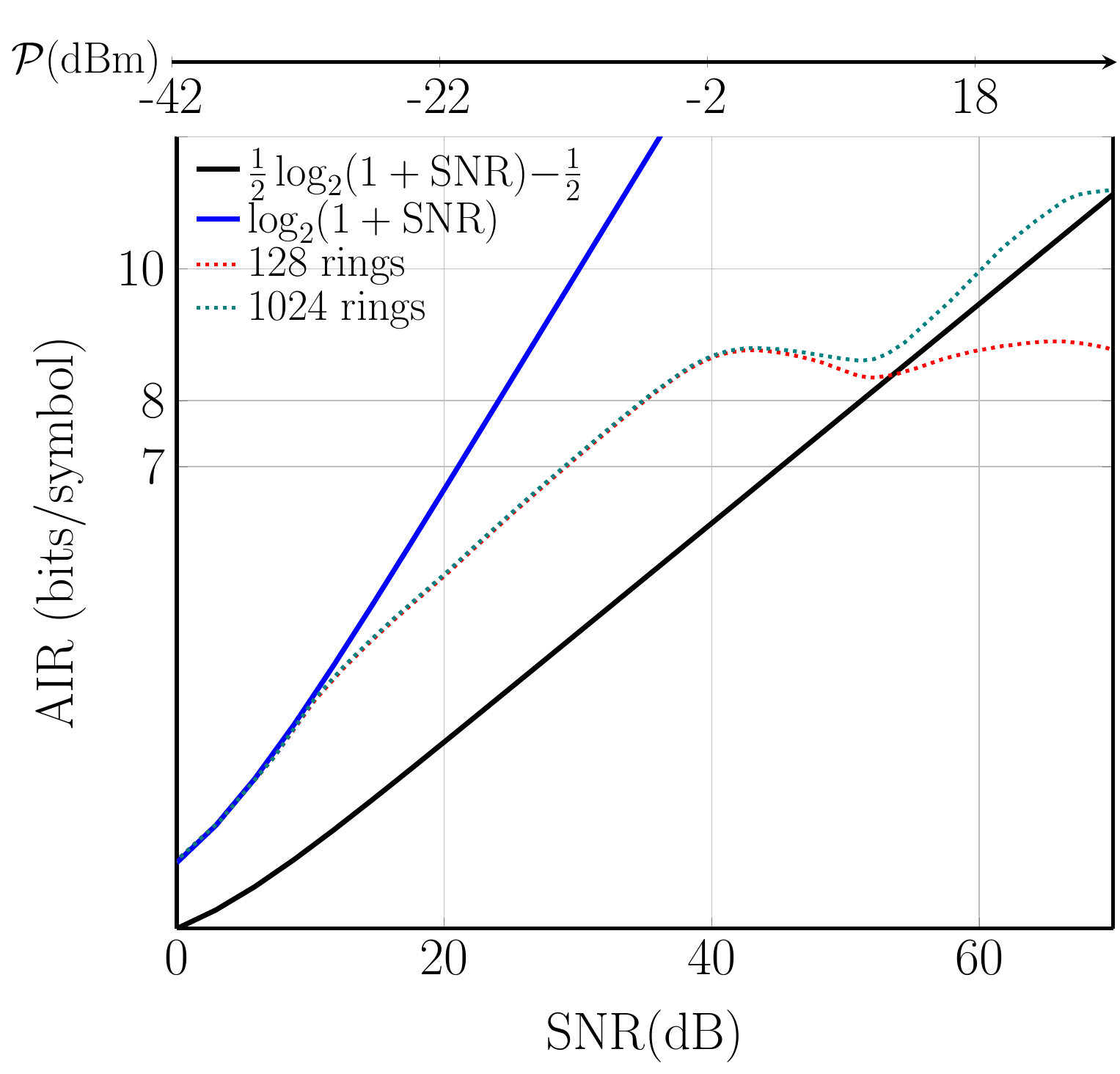}
\caption{$\mathcal{B}=5 \text{ GHz}$} 
 \label{fig:CapacityQAM16B5}
\end{subfigure}
\caption{AIR of the SSFM channel with large $K$ and back-propagation.}
 \label{fig:capacity}
\end{figure*}


\begin{figure}[t]
\centering
\begin{tabular}{cc}
\includestandalone[scale=0.5]{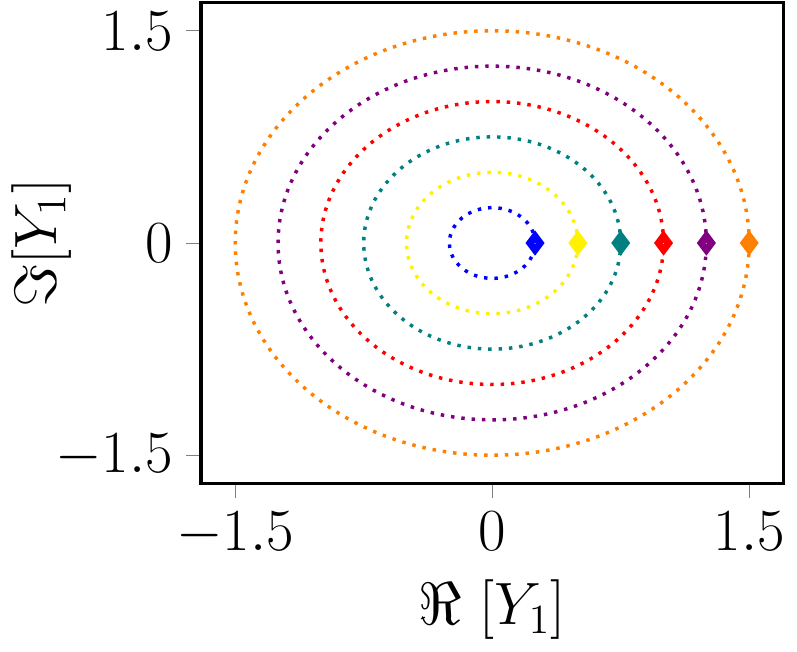}
&
\includestandalone[scale=0.5]{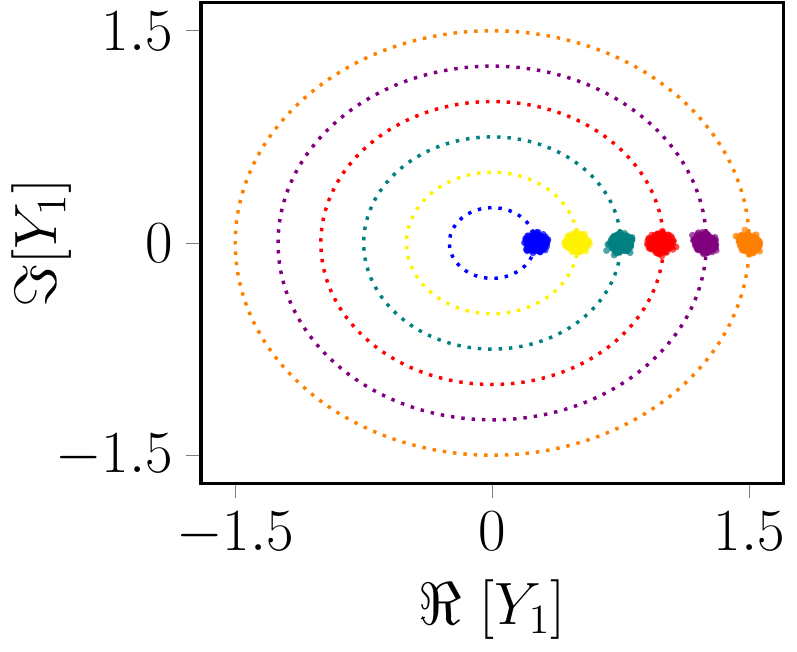}
\\
(a) Symbols at TX. & (b) RX, $\snr=30$ dB.
\\[4mm]
\includestandalone[scale=0.5]{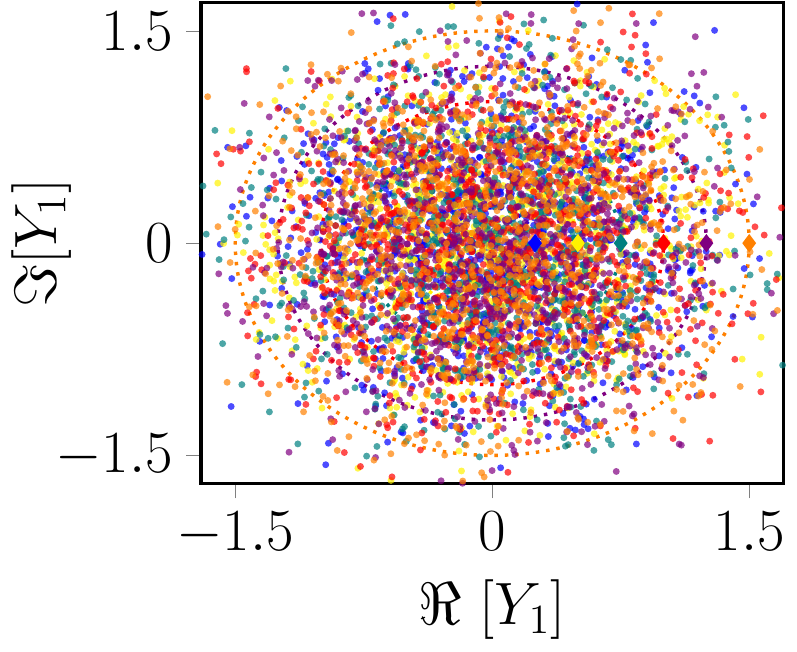}
&
\includestandalone[scale=0.5]{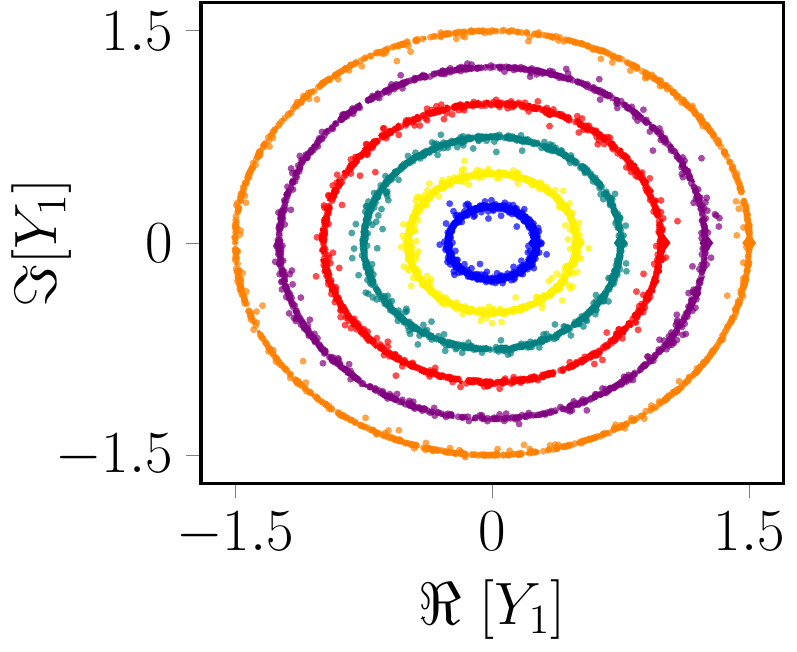}
\\
(c) RX, $\snr = 50$ dB. & (d) RX, $\snr=75$ dB.
\end{tabular}
\caption{Normalized constellation for the SSFM channel at $\mathcal{B}=20 \text{GHz}$. 
(a) Transmitted symbols. (b--d) Received symbols at several SNRs.
}
\label{fig:scatterNoise}
\end{figure}

\begin{figure*}[t]
\begin{subfigure}{\columnwidth}
\centering
  \includestandalone[scale=0.5]{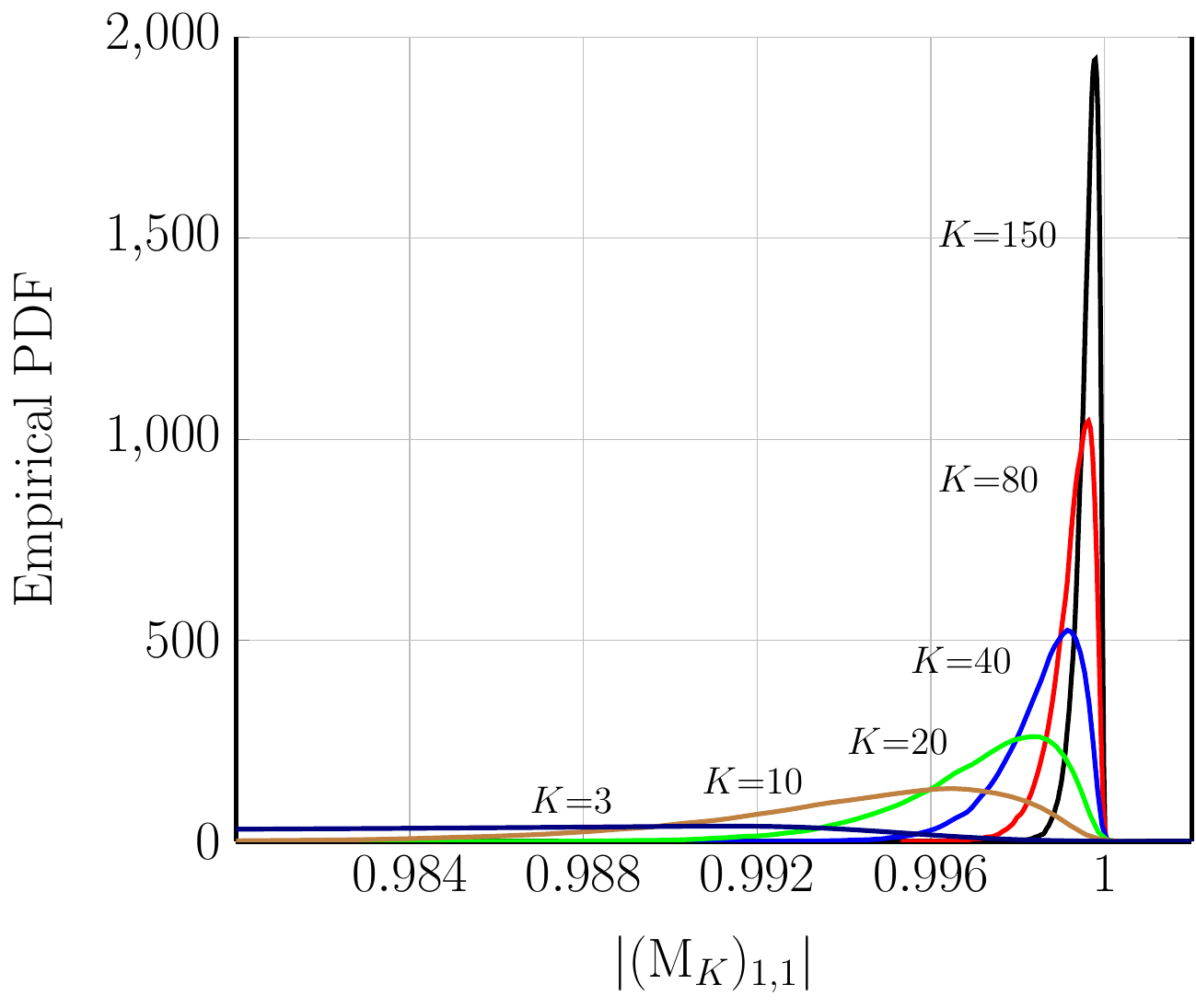}
\caption{Finite dispersion.} 
 \label{fig:divergenceFromoHaarOriginal}
\end{subfigure}
\begin{subfigure}{\columnwidth}
\centering
\includestandalone[scale=0.5]{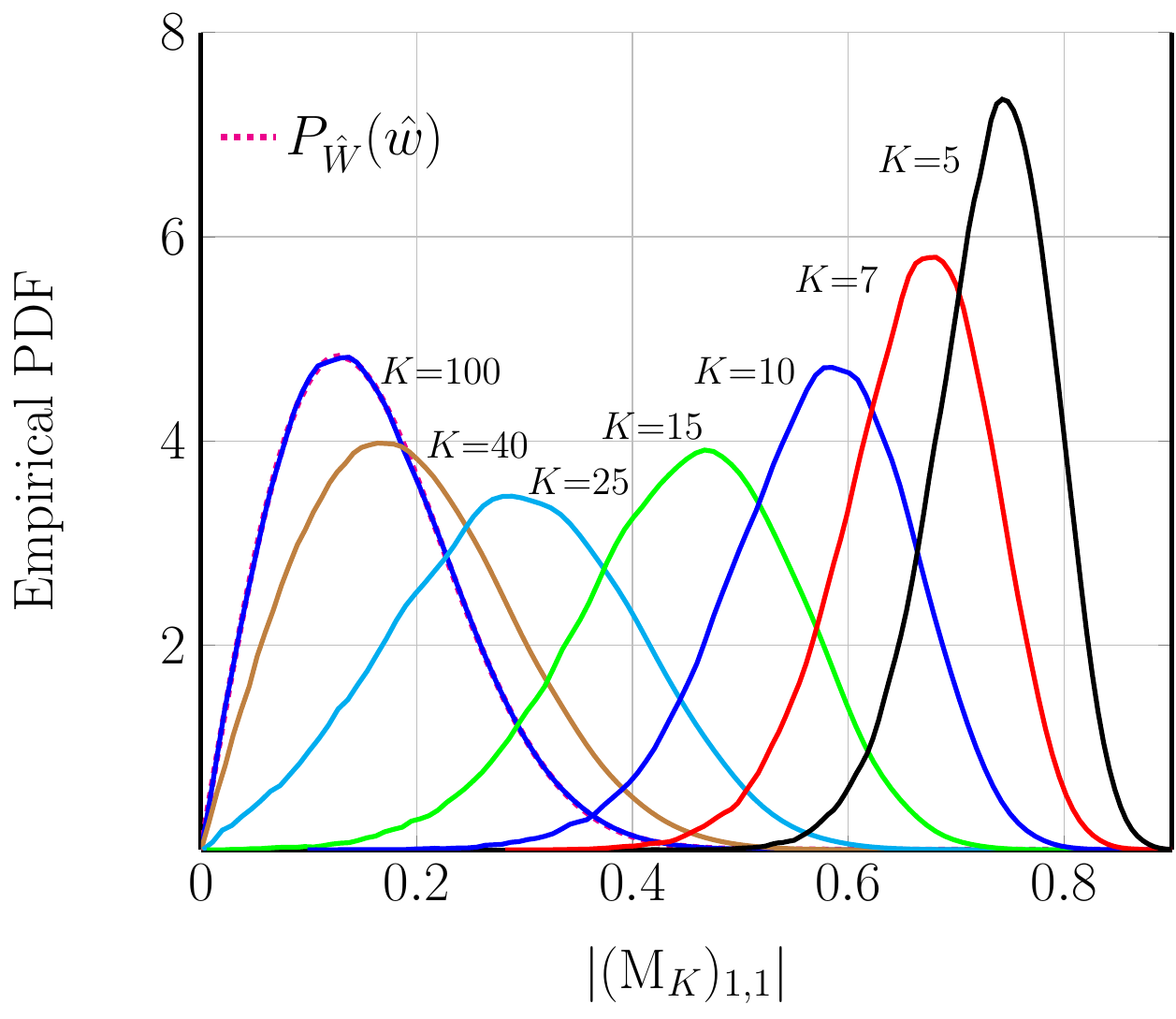}
\caption{Infinite dispersion.}
 \label{fig:convergenceToHaarInf}
\end{subfigure}
\caption{Empirical PDF of  $|(\mathsf{M}_K)_{1,1}|$ for the finite- and infinite- dispersion fading channel. }

\end{figure*}



\section{Capacity Simulation} 
\label{sec:simulations}

The capacity results in Section~\ref{sec:CapacityResultsExt} are supported by numerical simulation, presented in this section. 


We compute the maximum AIR by simulation, and compare that with the 
upper and lower bound \eqref{eq:ub} and \eqref{eq:mainThSSFM}. Furthermore, we investigate the properties of the random matrix $M_K$; in particular we demonstrate that $M_K$ tends to a diagonal 
matrix if $K$ is sufficiently large.

\subsection{Achievable information rate}

\begin{table}[t]
\caption{Fiber parameters}
\label{tab:params}
\centerline{\begin{tabular}{c|l|l}
$\alpha$ & $0.2$ dB/km & {\footnotesize fiber loss} \\
$D$      & 17 ps/(nm-km)         & {\footnotesize chromatic dispersion} \\
$\gamma$ & $1.27~{\rm W}^{-1}{\rm km}^{-1}$ & {\footnotesize nonlinearity parameter} \\
NF & 3 dB & {\footnotesize noise figure}\\
$h$ & $6.626 \times 10^{-34} {\rm J} \cdot {\rm s}$ & {\footnotesize Planck's constant} \\
$\lambda_0$ & $1.55~\mu{\rm m}$ & {\footnotesize carrier wavelength} 
\end{tabular}}
\end{table}

We consider an SSFM channel corresponding to a discretization of a fiber with parameters given in Tab.~\ref{tab:params},  $\mathcal{L}=2000 {\rm km}$, and $\mathcal{B}=20 \text{ GHz}$ and $\mathcal{B}=5 \text{ GHz}$, resulting in $\sigma^2=1.2\times 10^{-13} \text{ J/m}$ and $\sigma^2=3\times 10^{-14} \text{ J/m}$, respectively.  
We assume that fiber loss is perfectly compensated with distributed amplification, 
and choose time parameters $\Delta_t=1/\mathcal{B}$ and $n=4096$. 
Each element of the input vector is chosen i.i.d. from a uniformly-spaced multi-ring constellation with $m_A$ rings and 8 points in phase. 

The AIR is computed with equalization. Given output $\vc{Y}$, back-propagation is applied to obtain $\hat{\vc{Y}}$. The per-sample conditional PDF $p_{\hat{Y}_1|X_1}(\hat{y}_1|x_1)$ is numerically computed by averaging over all samples. The maximum of $I(X_1; \hat Y_1)$ over the input PDF provides a lower bound on the capacity
\begin{IEEEeqnarray}{rCl}
\mathcal{C}(\text{SNR})&\geq& \frac{1}{n} I(\vc{X};\hat{\vc{Y}}) \nonumber 
\\
&\geq& I(X_1; \hat Y_1),
\end{IEEEeqnarray}
where the last inequality holds for i.i.d. input. 

Fig. \ref{fig:capacity}(a) shows the maximum AIR  as a function  of the launch power and $\snr$  for $\mathcal{B}=20$ GHz.
It can be seen that the AIR is close  to the upper bound \eqref{eq:ub} in the low \snr\ regime $0\leq \snr\leq 15$ dB, and then, following a drop, increases again, approaching the lower bound \eqref{eq:mainThSSFM} as the \snr\ is increased. 

The AIR tends to infinity along the lower bound, which  appears to be tight in our simulations. 
Fig.~\ref{fig:capacity}(b) shows the convergence of the AIR to the lower bound at high powers for $\mathcal{B}=5$ GHz. Note that dispersion is stronger for larger bandwidth. As a consequence, the
stochastic ISI and the drop in the AIR are lower in Fig.~\ref{fig:capacity}(b) compared to those in Fig.~\ref{fig:capacity}(a).

Fig.~\ref{fig:scatterNoise} helps explain Fig. \ref{fig:capacity}(a), showing a number of symbols in the constellation at the transmitter (TX) and receiver (RX). 
In the regime $\snr <38\db $, the received symbols are localized around the transmitted symbols, and the AIR is between the upper and lower bounds \eqref{eq:ub} and \eqref{eq:mainThSSFM}.
In the medium \snr\ regime $45\db<\snr <55 \db$, the received symbols are almost independent of the transmitted symbol, resulting in almost zero AIR. 
Finally, in the high \snr\ regime $\snr>75\db$, the phase of the received symbols conditioned on the transmitted symbol is uniform; however the amplitude is now localized, limited by an additive ASE noise.
The AIR in this regime is $(1/2)\log_2(1+\snr)-1/2+o(1)$.

The analysis in Section~\ref{sec:CapacityResultsExt} shows that the SSFM model tends to a diagonal one for sufficiently large $K$ without equalization. Both deterministic and stochastic ISI tend to zero with $K$. 
Simulation of  the AIR without equalization shows a pattern similar to that in Fig. \ref{fig:capacity}, although
the value of the AIR is smaller due to deterministic ISI.

It follows that the AIR follows a double-ascent curve. As previously known, the AIR has an inverted bell curve shape, which corresponds to the range $\snr\leq 43$ dB in 
Fig. \ref{fig:capacity} (a). The existence of an optimal power in this range is attributed to a balance between the ASE noise and stochastic ISI. However, if \snr\ is further increased, the ISI eventually averages out to zero as proved in Lemma~\ref{lem:mainThSSFM-ex}. This gives rise to  the second ascent in the AIR, where the AIR approaches the rate of an interference-free phase noise channel. 

Note that equalization using back-propagation improves the AIR at low-to-medium SNRs by canceling the deterministic component of the inter-symbol inference (ISI). At high SNRs there is no benefit in applying equalization as the model is already ISI-free. However, if equalization is applied, the model remains diagonal since the phase at the input of the equalizer is uniform conditioned on the channel input.

\subsection{Conditional PDF in the Fading Channel}

In this and the next section, we verify  the properties of the 
conditional PDF in the fading and SSFM channels. As the input is
multi-dimensional, we compute the distribution of specific entries of the channel matrix $\mathsf M$.

In the first experiment, we simulate the random matrix $\mathsf M_k$ \eqref{def:matrixMk} for finite dispersion case with zero loss, $n=32$, 
and  for values of $b_{\ell}$ in \eqref{def:dispersions} with $T  =50$, $\beta_2=-2$, and $\mathcal{L}=1/4$.   Fig~\ref{fig:divergenceFromoHaarOriginal} shows that the empirical PDF of $W \triangleq |(\mathsf{M}_K)_{1,1}|$ converges to the Dirac delta function $\delta(W-1)$, which is explained by Lemma~\ref{lem:asymDistK}. For these choices of parameters, $b_{\ell}$ are small and the  PDF of $|(\mathsf{M}_K)_{1,1}|$ tends to $\delta(W-1)$ as $K$ increases.

In the second experiment, the random matrix $\mathsf M_k$ \eqref{def:matrixMk} is simulated for infinite dispersion case with zero loss and $n=32$. The values of $b_{\ell}$ are chosen to be numbers in the interval $(0,\pi/3]$ such that the matrix $\mathrm{D}$ becomes a non-block diagonal matrix. Fig~\ref{fig:convergenceToHaarInf} shows that the empirical PDF of $W \triangleq |(\mathsf{M}_K)_{1,1}|$ converges to the PDF of $\hat{W}=|\mathsf{M}_{1,1}|$, where $\mathsf{M}$ is distributed according to the Haar measure over the group of unitary matrices. This supports the result of Lemma~\ref{lem:haarMeasure}. Note that, by Theorem~\ref{th:pdfFirstColumn} in Appendix~\ref{sec:math},  $P_{\hat{W}}(\hat{w})=2(n-1)\hat{w}(1-\hat{w}^2)^{n-2}$. 

In the third experiment, the first simulation is repeated with the same parameters except with $\mathcal{L}=25$, which results in larger absolute values for $b_{\ell}$. In this case, as $K$ is increased, the empirical PDF of $W$ first, very fast as shown in Fig.~\ref{fig:firstConvegenceToHaar}, gets close to the PDF of $\hat{W}$, which by Lemma~\ref{lem:haarMeasure} corresponds to the PDF of $|(\mathsf{M}_K)_{1,1}|$ in infinite dispersion fading channel when $K \ra \infty$. Hence, it seems that when $|b_{\ell}|$ are not small, the PDF of $|(\mathsf{M}_K)_{1,1}|$ in the fading channel with finite dispersion is similar to the infinite dispersion case. 

As $K$ is further increased, the distribution of $\mathsf{M}_K$ gets far from PDF of $\hat{W}$; as a consequence,
as it can be observed in Fig.~\ref{fig:divergenceFromoHaar}, the PDF of $W$ tends to $\delta(W-1)$ rather than \eqref{eq:PW1}.

\begin{figure*}
\begin{subfigure}{\columnwidth}
\centering
    \includestandalone[scale=0.5]{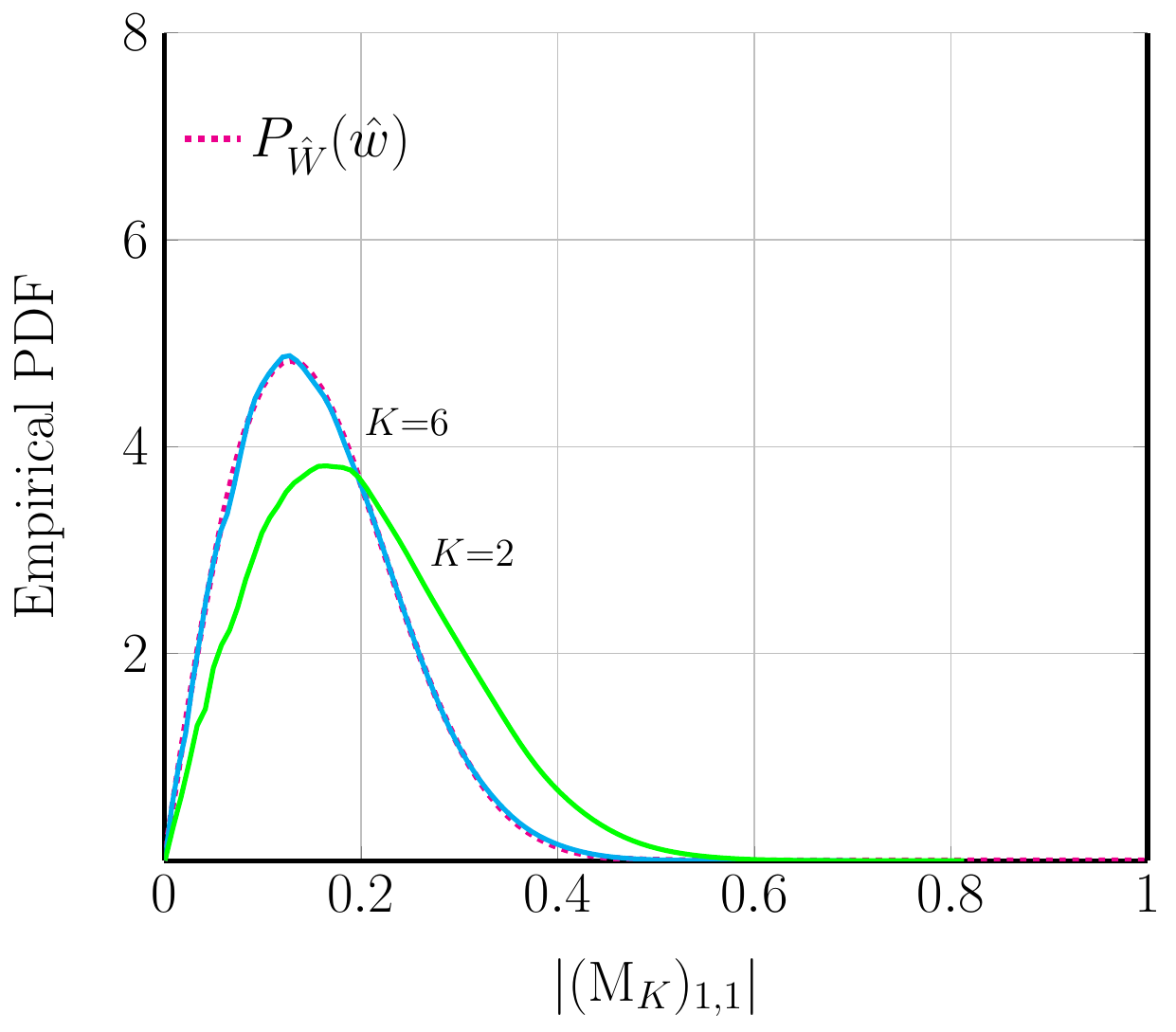}
\caption{First, PDF gets  close to Haar distribution.}
  \label{fig:firstConvegenceToHaar}
\end{subfigure}
\begin{subfigure}{\columnwidth}
\centering
  \includestandalone[scale=0.5]{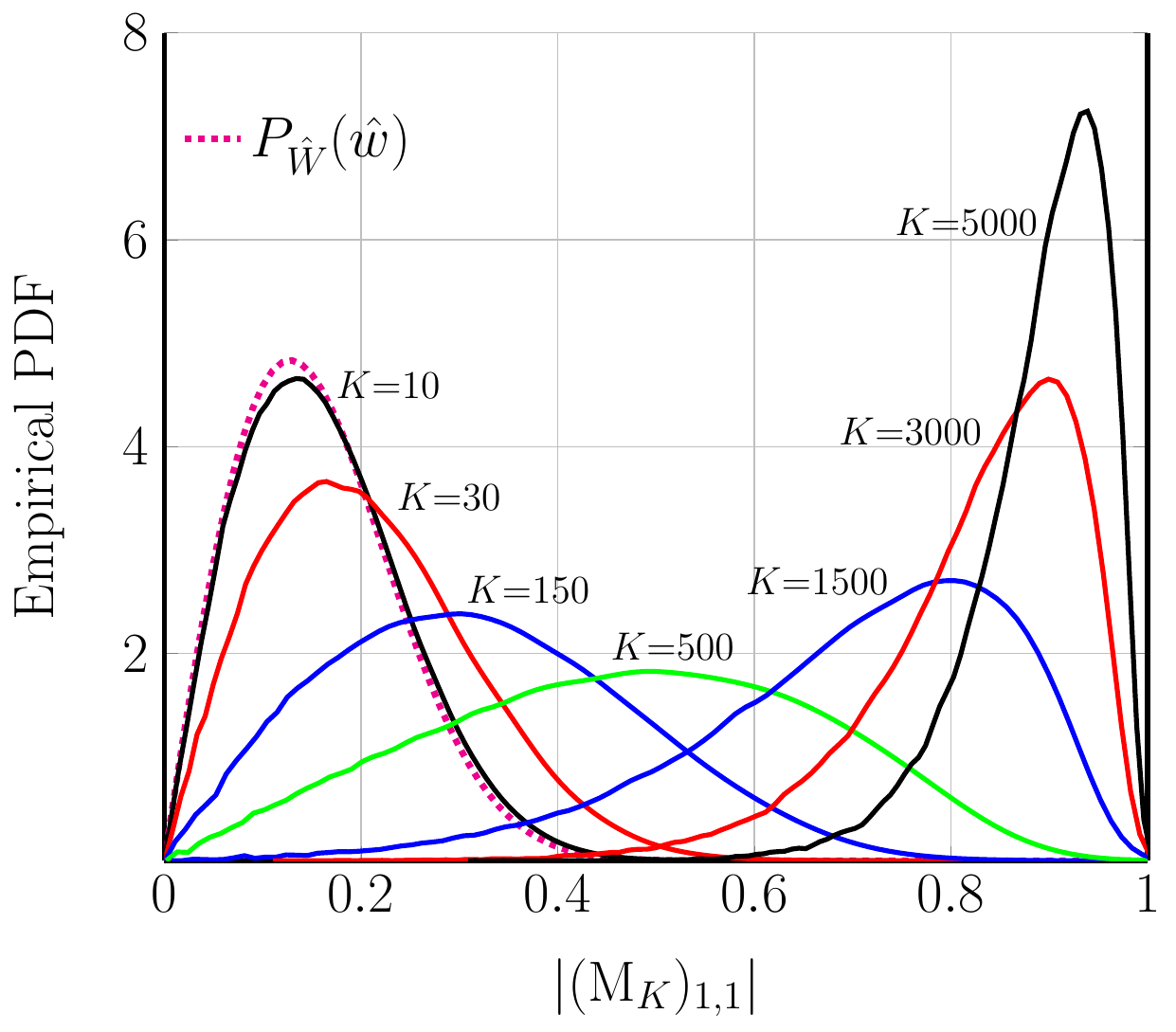}
\caption{Then, PDF converges to $\delta(w-1)$.}
  \label{fig:divergenceFromoHaar}
\end{subfigure}

\caption{Empirical PDF of $|(\mathsf{M}_K)_{1,1}|$ for finite dispersion.}
\label{fig:exampleFadingHaarDirac}
\end{figure*}

\begin{figure}[t]
  \centering
\includestandalone[scale=0.5]{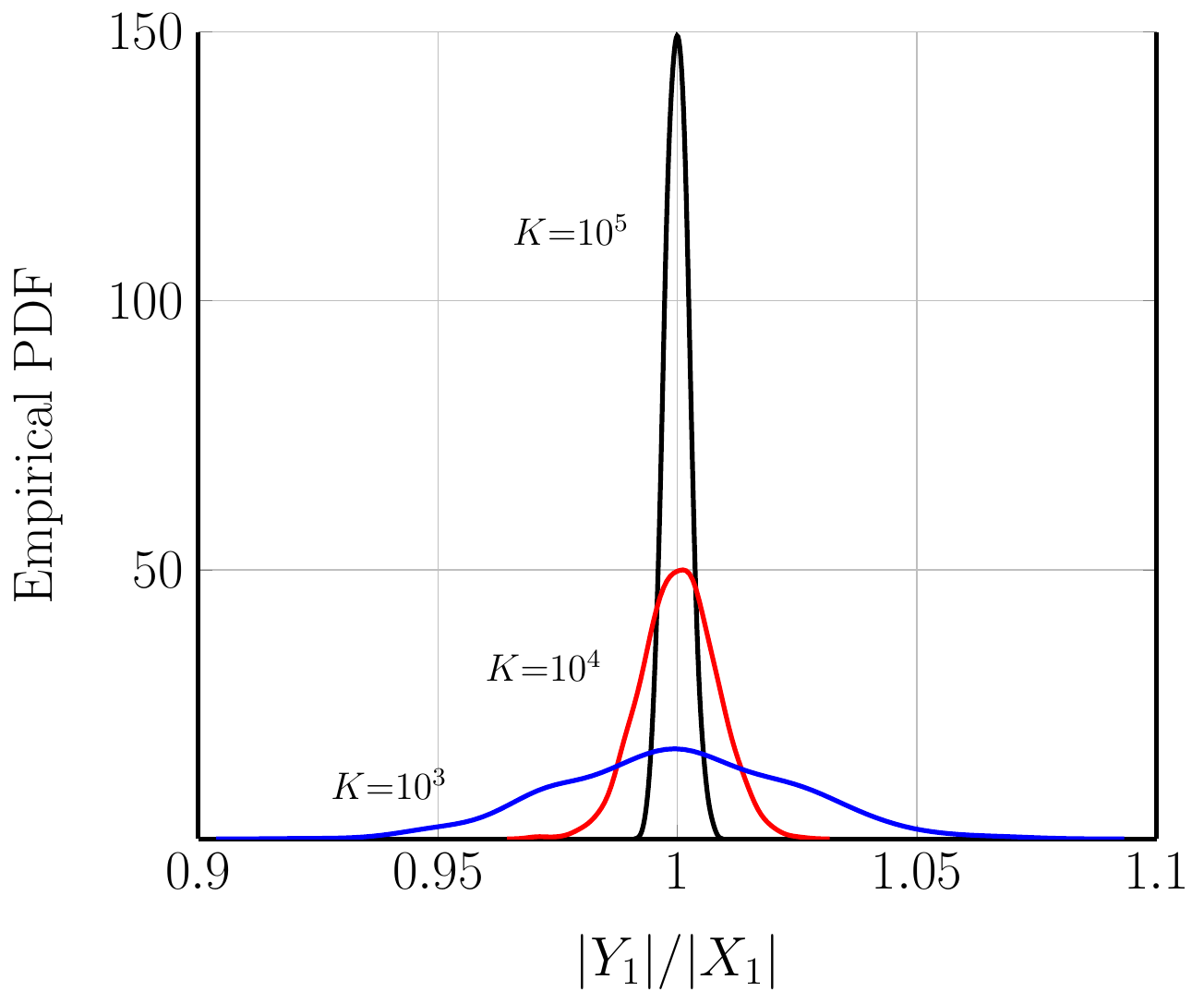}
\caption{Empirical PDF of $|Y_1| /|X_1|$ for SSFM channel with $n=32$.} 
 \label{fig:SSFMN32}
\end{figure}


\subsection{Conditional PDF in the SSFM Channel}
We verify that the SSFM channel is nearly diagonal when $K$ is sufficiently large.
In the first experiment,  we simulate a lossless channel with 1000 realizations of the noise and large input  $X_{\ell}\widesim{\text{i.i.d.}}10^8(\mathcal{U}(0,1)+0.7)$. We compute the empirical PDF of $|Y_1|/|X_1|$ to show it converges to the Dirac Delta function. Similarly this holds for any $i\in [n]$. 
 The number of spatial segments $K$ is chosen as follows.
 Considering the proof of Lemma~\ref{lem:asymDistK} in Section \ref{sec:proof}, the SSFM model is diagonal when $\mathrm{C}_1/\sqrt{K}$ is small,  where the matrix $\mathrm{C}_1$ is defined in \eqref{eq:C-i}. Hence,  $\frac{1}{\sqrt{K}}\max\limits_{\ell}|d_{\ell}|$ should be small, \eg\ less than $0.1$.  For the normalized NLS equation in \cite{yousefi2012nft1}, $\max \limits_{\ell}|d_{\ell}|=\left(n \pi/T\right)^2$. Letting $T=50$, we obtain $K \geq 16(n/10)^4$.

\paragraph*{Example n=32} In this case, $K > 1500$. Fig.~\ref{fig:SSFMN32}, shows that the empirical PDF of $|Y_1|/|X_1|$ is concentrated around 1, with $\sigma^2=5\times 10^{-5}$. This supports the relation $|Y_1|=|X_1|$ with probability one as $K\ra\infty$.

\paragraph*{Example n=1024} In this case, $K >  1.7\times 10^9$ and simulation is infeasible. However, if we reduce  $d_{\ell}$ by factor 100 (or 400), we obtain $K > 10^4$ (or $K > 1.7\times 10^5$). Fig.~\ref{fig:SSFMN1024D400} shows the empirical PDF of $|Y_1|/|X_1|$ for  several values of $K$ and  $\sigma^2=1.5\times 10^{-3}$, demonstrating $|Y_1| \approx |X_1|$.

 \begin{figure*}[ht]
\begin{subfigure}{\columnwidth}
  \centering
  \includestandalone[scale=0.5]{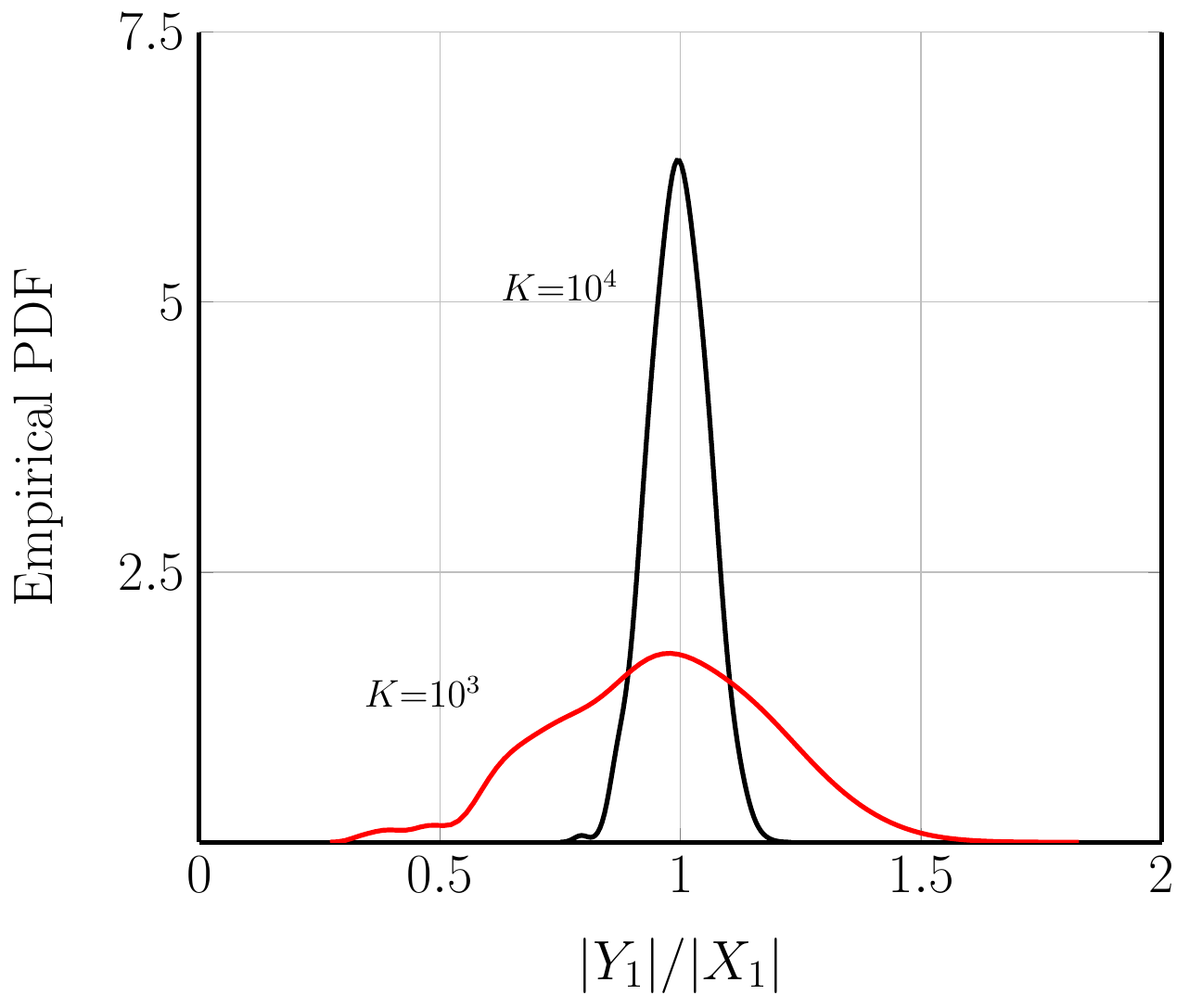}
\caption{Dispersion  values $b_{\ell}$ divided by 100, $\forall l$. }
  \label{fig:SSFMN1024D100}
\end{subfigure}
\begin{subfigure}{\columnwidth}
  \centering
\includestandalone[scale=0.5]{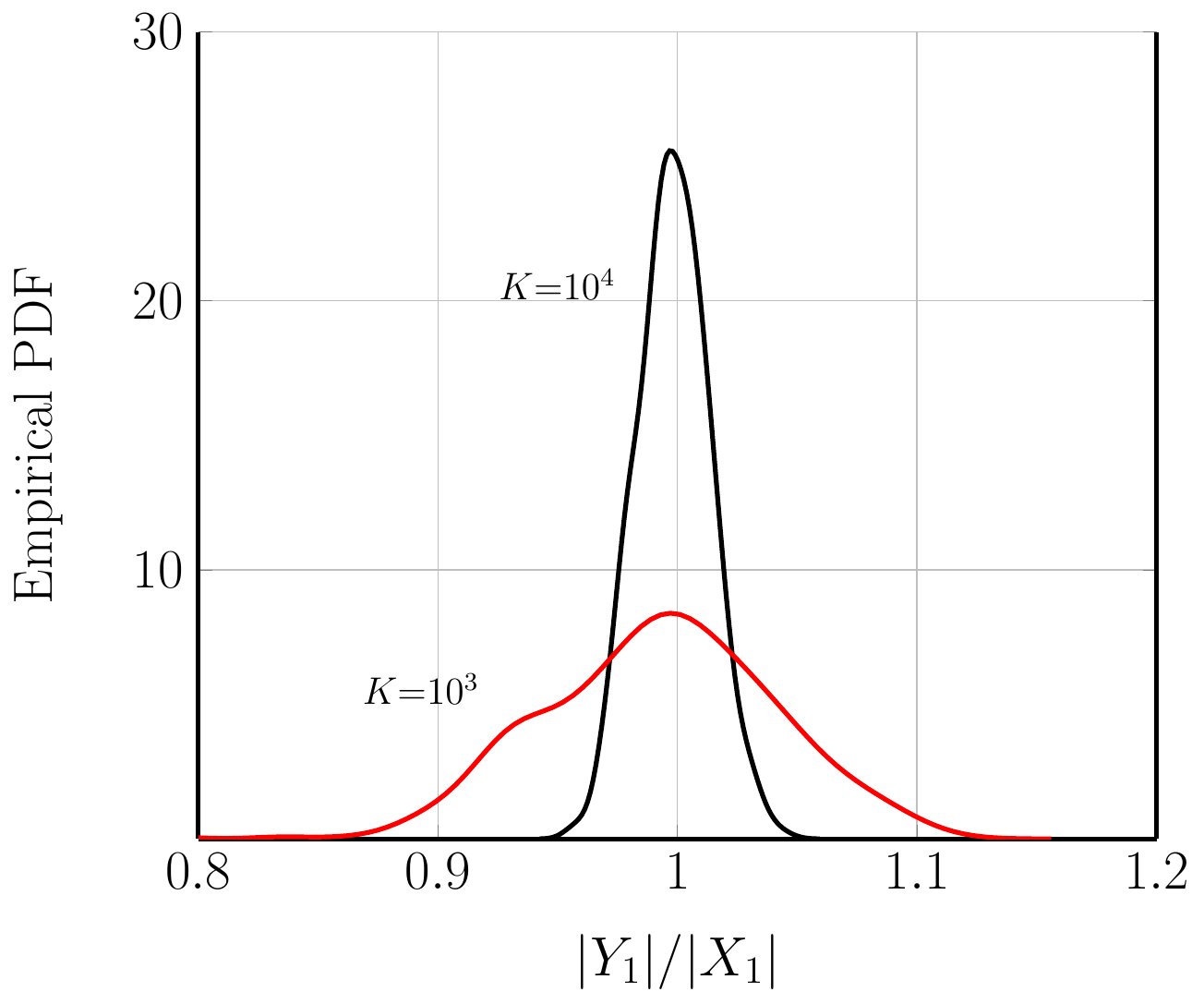}
\caption{Dispersion values $b_{\ell}$ divided by 400, $\forall l$.}
  \label{fig:SSFMN1024D400}
\end{subfigure}
\caption{Empirical PDF of $|Y_1|/|X_1|$ for SSFM channel with $n=1024$ and dispersion values divided by 100 and 400.}
\label{fig:fig:SSFMN1024D}
\end{figure*}

In the second experiment, we  investigate the rate of convergence of the SSFM channel to the diagonal phase noise model. We consider the normalized SSFM with $n=32$, $K=10^4$, $\sigma^2=5\times 10^{-5}$, and 1000 realizations of the noise and input  $X_{\ell}\widesim{\text{i.i.d.}}\sqrt{\frac{\mathcal P}{1.5}}(\mathcal{U}(0,1)+0.7)$, where $\mathcal{P}=K^{\delta}$, for $0 \leq \delta \leq 4$.

From Lemma~\ref{lem:mainThSSFM-ex} and  \eqref{eq:phiGeneral}, the rate of convergence of $\mathsf{M}_K$ to $e^{\zeta+jd}\mathsf{S}_K$ is
\begin{IEEEeqnarray}{c}
    \underline{\upsilon}(\delta)\triangleq-\log_K\left(\frac{1}{K}\sum \limits_{i=1}^{r}e^{j \sum \limits_{\ell=i}^K \left(\Phi_{\ell,1}\right)}\right).
\end{IEEEeqnarray}
In Fig.~\ref{fig:orderConvergence},  $\underline{\upsilon}(\delta)$ is simulated. The results are compatible with  Lemma~\ref{lem:mainThSSFM-ex} stating $\underline{\upsilon}(\delta) > {\upsilon}(\delta)-\epsilon'$ for any $\epsilon'>0$, where $\upsilon(\delta)$ is defined in \eqref{def:upsilonDelta}. Moreover, it can be seen in the proof of Lemma~\ref{lem:mainThSSFM-ex} that for small values of $\sigma^2\mathcal{L}$, the signal-noise mixing may not be dominant except for very large  $K \gg \left(\gamma \mathcal{L}^{3/2} \sigma\right)^{-2/\delta}$.

When $\gamma \mathcal{L}^{3/2} \sigma K^{\delta/2}$ is small,  $\underline{\upsilon}(\delta)$ can be lower bounded as
\begin{IEEEeqnarray}{c}
       \underline{\upsilon}(\delta)+\epsilon'\geq \begin{cases}
       \IEEEstrut
        \delta,& 0 \leq \delta \leq 1,\\
        1.5-\delta/2-g,& 1 \leq \delta \leq 2  \;\; \\
        &\text{and i.i.d. input}, \\
         0.5,& 2 \leq \delta \leq 3 \;\; \\
        & \text{and i.i.d. input},\\
        0.5,& 3 \leq \delta,
        \IEEEstrut
\end{cases} \label{def:upsilonDeltaLowrange}
\end{IEEEeqnarray}
which is compatible with the simulation result in Fig.~\ref{fig:orderConvergence}. The oscillation for $1\leq \delta \leq 2$ is also explained by the vanishing oscillating term $1/(e^{jK^{\delta-1}Q}-1)$ in \eqref{eq:oscillation} in the proof of Lemma~\ref{lem:mainThSSFM-ex}. 

In the last experiment, the effect of loss is examined. The previous experiment is repeated with fixed $\delta=0.6$ and $\zeta=-1.35$. For $K=100$, $K=1000$, and $K=1000$, the convergence rate $\underline{\upsilon}(\delta)$ are is $0.463$, $0.512$, and $0.534$, respectively. This is explained by \eqref{eq:lossEffect}, implying that for $0<\delta<1$ and small noise power,  $\underline{\upsilon}(\delta)=\delta+2a\zeta/\ln{K}\ra \delta$, where $a$ is a value less than 1. In our experiment, $a\approx 0.23$.


\section{Conclusion} \label{sec:conc}

The capacity of the discrete-time SSFM model of  optical fiber is considered as a function of the average input signal power $\snr$, when the number of spatial segments in SSFM  is sufficiently large as $K = \sqrt[\delta]{\snr}$, $\delta>0$.

First, we obtained the capacity lower bound \eqref{eq:capacity-a} and characterized the pre-log as a function of $\delta$.
In particular, we showed that  $\mathcal{C}(\snr)\geq \frac{1}{2}\log_2\left(1+\snr\right)- \frac{1}{2}+o(1)$, where $o(1)$ vanishes as $\snr \ra \infty $.  
As a result, the number of signal DoFs is at least half of the input dimension. 

Second, it is shown that the
capacity of the continuous-space SSFM channel when $\gamma \ra \infty$ is $\frac{1}{2}\log_2(1+a\, \snr)+o(1)$. Hence, the number of signal DoFs is exactly half of the input dimension. 

Third, we considered the SSFM model, named as infinite-dispersion,  where the dispersion matrix in each segment does not depend on $K$. 
It is shown that if $K= \sqrt[\delta]{\snr}$, $\delta>3$,  then $\mathcal{C}(\snr,\sqrt[\delta]{\snr}) = \frac{1}{2n}\log_2(1+\snr)+O(1)$, where $O(1)$ term is bounded as $\snr \ra \infty$. Here, there is exactly one signal DoF.

Finally, AIRs of the SSFM model with back-propagation equalization are obtained numerically. The results show that while the AIR drops significantly in the medium $\snr$  regime due to a considerable stochastic ISI, it asymptotically converges to $\frac{1}{2}\log_2\left(1+\snr\right)- \frac{1}{2}+o(1)$, explained by the fact that ISI vanishes at high \snr s.


\appendices
\section{Mathematical Preliminaries} \label{sec:math}
\subsection{Spherical Coordinate System}
The $n$-dimensional spherical coordinate system is described by a radius, $r$, and $n-1$ angles $\theta_{\ell}$, $\ell\in[n-1]$, where $\theta_{\ell}\in[0,\pi]$ for $\ell\in[n-2]$, and $\theta_{n-1}\in[0,2\pi)$. A vector $\vc{x}\in \mathbb{R}^n$ can be written in the spherical coordinate  as
\begin{IEEEeqnarray}{rcl}
x_1&=&r\cos(\theta_1), \nonumber\\
x_{\ell}&=&r\cos(\theta_{\ell})\prod \limits_{r=1}^{\ell-1} \sin(\theta_{r}) ,~\ell = 2,\ldots,n-2 \nonumber\\
x_{n-1}&=&r\prod \limits_{r=1}^{n-1} \sin(\theta_{r}).
\end{IEEEeqnarray}

Alternatively, we denote $\vc{x}$ by its norm $\norm{\vc{x}}$ and  direction $\hat{\vc{x}}=\vc{x}/\norm{\vc{x}}$ on the surface of the $n-1$-sphere
\begin{IEEEeqnarray}{c}
    \mathcal{S}^{n-1}=\left\{\vc{x} \in \mathbb{R}^n: \norm{x}=1\right\}.
\end{IEEEeqnarray}

Complex vectors in $\mathbb{C}^n$ can be similarly represented.

\subsection{Groups}

The reader is referred to \cite{Quint14,Breuillard04,Stromberg60} for background on group theory. 
For a group $G$, notation $H \leq G$ is used to say that $H$ is a subgroup of $G$  and $gH$ and $Hg$ are respectively the left coset and right coset of $H$ w.r.t. $g\in G$.

A probability measure on $G$ is a non-negative, real-valued, countably additive, regular Borel measure $\mu$ on $G$, such that $\mu(G)=1$. The support of $\mu$, denoted by $\mathcal{S}({\mu})$ is the smallest closed subset of $G$ of $\mu$-measure.

A probability measure $\mu$ on $G$ is said to be (normally) aperiodic if its support is not contained in a (left or right) coset of a proper closed (normal) subgroup of $G$. 

\paragraph{Group of Unitary Matrices} The group that we are interested in this paper is the group of unitary matrices. A  matrix $\mathrm{U}\in\mathbb{C}^{n\times n}$ is unitary if 
\begin{IEEEeqnarray}{c}
    \mathrm{U}  \mathrm{U}^H= \mathrm{U}^H  \mathrm{U}=\mathrm{I}_n,
\end{IEEEeqnarray}
where $\mathrm{U}^H$ denotes the conjugate transpose of $\mathrm{U}$. The set of unitary matrices in $\mathbb{C}^{n\times n}$ with matrix multiplication forms a group $\mathbb{U}_n$, which is a compact Lie group.

The following theorem is a re-statement of \cite[Thm.~1]{Borevich81}, bringing parts of its proof to the theorem statement.

\begin{theorem} \label{th:Borevich}
Suppose that a subgroup $\mathbb{H} \leq \mathbb{U}_n$ contains the subgroup of diagonal matrices. 
Let $\nu$ be a binary relation on $\mathcal{I}_n=\{1,2,\ldots,n\}$ defined as follows: $\forall r,s\in[n]$,  $r \widesim{\nu} s$ if and only if there exists a matrix $\mathrm{A} \in \mathbb{H}$ and 
$ t \in[n]$ such that $\mathrm{A}_{r,t} \neq 0$ and $\mathrm{A}_{s,t} \neq 0$. 

Then, we have
\begin{itemize}
\item[i.] $\nu$ is an equivalence relation on $\mathcal{I}_n$,
\item[ii.]$\mathbb{U}_n(\nu) \leq \mathbb{H}$.
\end{itemize}
\end{theorem}

\subsection{Haar Measure}

This subsection is borrowed mainly from \cite{Meckes04}. Haar measure can be seen as an extension of the notion of the uniform random variable over an interval. 
The extension is based on the shift-invariant property of the uniform random variable. If $X\sim \mathcal{U}(0,a)$, then for any $b\in\mathbb R$, $\mod(X+b,a)\sim \mathcal U(0,a)$. 

Consider defining uniform distribution on the circle $\mathcal{S}^1$ in $\mathbb R^2$. Considering a circle as a geometric object, a “uniform random point on the circle” should be a complex random variable whose distribution is rotation invariant;  that is, if $A\subseteq \mathcal{S}^1$, then the probability of the random point lying on $A$ should be the same as the probability that it lies on $e^{j\theta}A=\{e^{j\theta} a:a\in A\}$.

The uniform distribution $\mu$ on a group $G$ is called Haar measure on $G$, defined based on the``translation-invariant'' property as follows. 
For a group $(G,\cdot)$, an element $g \in G$, and a Borel subset $\mathcal{S} \subseteq G$, the left translation of $\mathcal{S}$ by g is defined as
\begin{IEEEeqnarray}{c}
    g\mathcal{S}=\{g \cdot s: s \in \mathcal{S}\}.
\end{IEEEeqnarray}
A measure $\mu$ on the Borel subsets of $G$ is called left translation-invariant if for all Borel subsets $\mathcal{S} \subseteq G$ and all $g \in G$,
\begin{IEEEeqnarray}{c}
    \mu(g\mathcal{S})= \mu (\mathcal{S}).
\end{IEEEeqnarray}
Right translation and right translation-invariant are defined similarly.

There exists a unique Haar measure on any group. 
The following theorem is proved in \cite[Lemma~2.1.]{Meckes04} for $G=\mathbb U_n$.

\begin{theorem} \label{th:existenceHaar}
There exists a unique  (left or right) translation-invariant probability measure on $\mathbb{U}_n$, called Haar measure.
\end{theorem}

Let $\mathsf{M}\in\mathbb C^{n\times n}$ be a random unitary matrix, $W_{\ell} \triangleq |\mathsf{M}_{\ell,1}|$, and $\vc W= (W_1, \ldots, W_n)$.

\begin{theorem} \label{th:pdfFirstColumn}
 Suppose that $\mathsf{M}$ is distributed according to Haar measure on the group of random unitary matrices. Then,
\begin{IEEEeqnarray}{c}
P_{\vc{W}}(\vc{w})=2^{n-1}(n-1)! \prod \limits_{\ell=1}^{n-1}  w_{\ell}.
\label{eq:PW}
\end{IEEEeqnarray}
\end{theorem}

\begin{proof}
Let $W_{\ell}=R_{\ell}+jT_{\ell}$, $R_{\ell},T_{\ell} \in \mathbb{R}$, $\ell \in [n]$. The vector
\begin{IEEEeqnarray*}{c}
\begin{pmatrix}
R_1,& \ldots,&R_n, & T_1,&  \ldots, &T_n
\end{pmatrix}
\end{IEEEeqnarray*}
has the uniform distribution over $S_{2n-1}$. From \cite[Eq.1.26.]{FangWangng18}, the joint distribution of $(R_1,T_1)$ is
\begin{equation*}
P_{R_1,T_1}(r_1,t_1)=\frac{\Gamma(n)}{\Gamma(n-1) \pi}\left( 1- (r_1^2+t_1^2)\right)^{n-2}.
\end{equation*}

Since the phase is uniform, we get
\begin{equation*}
P_{W_1}(w_1)=2(n-1)w_1 \left( 1- w_1^2\right)^{n-2}.
\end{equation*}
Similarly, the joint distribution of $(R_1^2,T_1^2)$ is
\begin{equation*}
P_{R_1^2,T_1^2}(r_1^2,t_1^2)=
\frac{\Gamma(n)}{\Gamma(n-2)\pi^2}\left( 1- (r_1^2+t_1^2+r_2^2+t_2^2)\right)^{n-3}.
\end{equation*}
Again since phase is uniform, then
\begin{IEEEeqnarray*}{rcl}
P_{W_1^2}(w_1^2)&=&2^2(n-1)(n-2)w_1w_2 \left( 1- (w_1^2+w_2^2)\right)^{n-3}.
\end{IEEEeqnarray*}
The result can be similarly established for $W_1^n$ by induction.
\end{proof}

The joint PDF \eqref{eq:PW} gives the marginal $P_{W_1}(w_1)=g(w_1,n)$, where
 \begin{IEEEeqnarray}{c}
     g(w,n)\triangleq 2(n-1)w(1-w^2)^{n-2}.
 \end{IEEEeqnarray}
The conditional PDFs for $\ell=2,\ldots, n$ are 
\begin{IEEEeqnarray}{rCl}
P_{W_{\ell}|W_1^{\ell-1}}(w_{\ell}|w_1^{\ell-1})&=&
\Bigl(1-\sum \limits_{i=1}^{\ell-1}w_{i}^2\Bigr)^{-\frac{1}{2}} \label{eq:PW1}\\
&&
\times g\Bigl(
\Bigl(1-\sum \limits_{i=1}^{\ell-1}w_{i}^2\Bigr)^{-\frac{1}{2}} w_{\ell}
,n-\ell+1\Bigr).
\nonumber
\end{IEEEeqnarray}
\begin{figure}[t]
  \centering
  \includestandalone[scale=0.5]{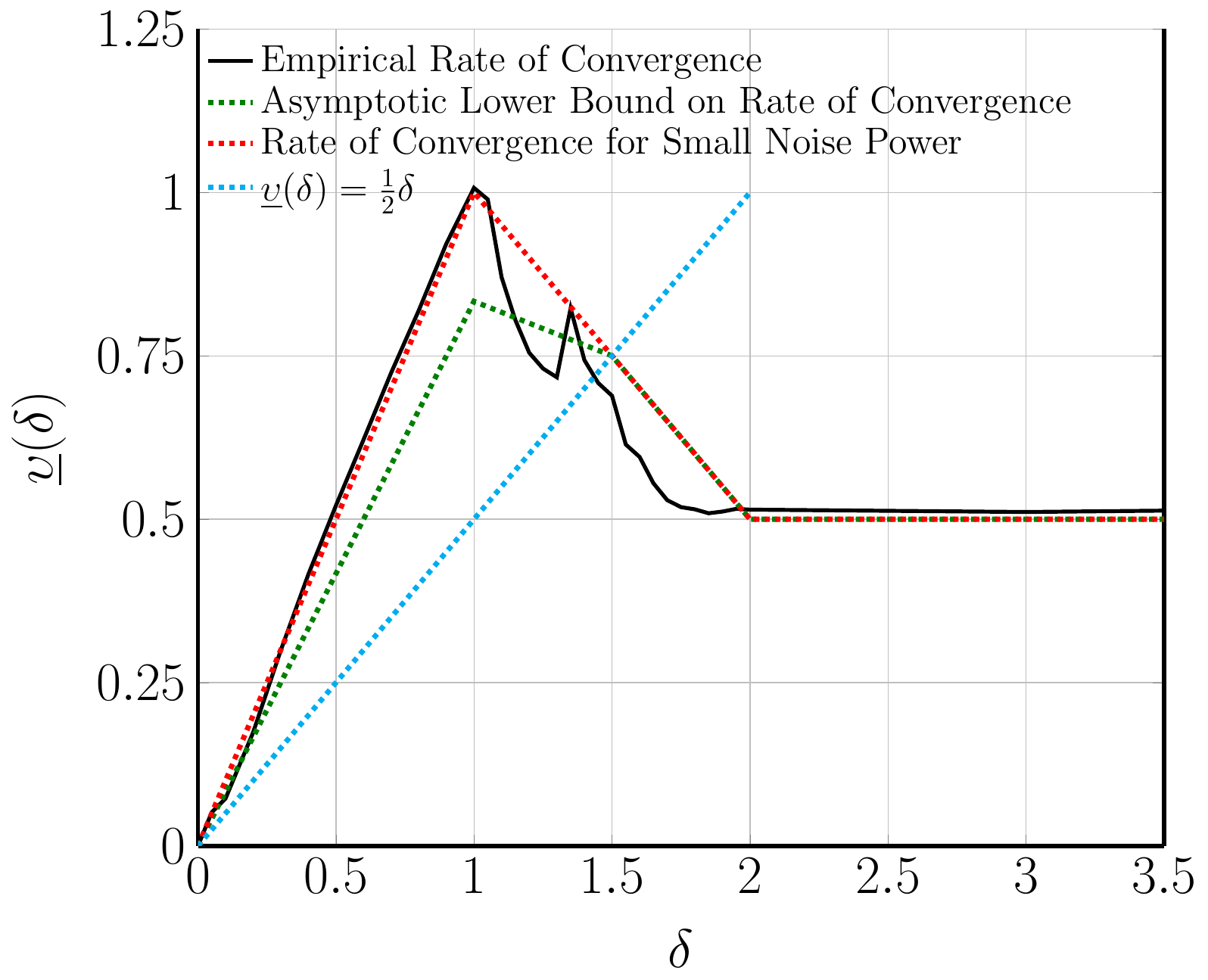}
\caption{Rate of Convergence of SSFM channel to the diagonal model for $n=32$, $K=10^4$, and $\mathcal{P}=K^{\delta}$.}
  \label{fig:orderConvergence}
\end{figure}

\begin{theorem} \label{th:independence}
 Suppose that $\mathsf{M}\in\mathbb C^{n\times n}$ is distributed according to the Haar measure on $\mathbb U_n$ and $\vc{x} \in \mathbb{C}^n$ with $\norm{\vc{x}}=1$. Then, $P_{\mathsf{M}\mathbf{X}|\mathbf{X}}\left(\mathsf{M}\mathbf{x}|\mathbf{x}\right)$ is independent of $\mathbf{x}$.
\end{theorem}

\begin{proof}
Fix an orthonormal basis $V_1,\ldots,V_n$ of $\mathbb C^n$  such that $V_1=\mathbf{x}$. Denote the matrix with columns $V_i$ by $\mathrm{\Lambda}$. Assume that $\mathsf{M}$ is a unitary matrix distributed according to Haar measure. Define the map 
$\Phi: \mathbb{C}^{n\times n} \mapsto \mathbb{C}^{n\times n}$ as
\begin{IEEEeqnarray}{c}
\mathsf{M}'=\mathsf{M}\Lambda.
\end{IEEEeqnarray}
Thus, 
\begin{IEEEeqnarray}{c}
P(\mathsf{M}\mathbf{x}|\mathbf{x})=P(\mathsf{M}'\mathbf{x}'|\mathbf{x}'=\left(1,0,\ldots,0\right)^T),
\end{IEEEeqnarray}
Since $\Phi$ is invertible, then $\mathsf{M}'$ is also distributed according to Haar measure and 
 the RHS of above is derived in Theorem~\ref{th:pdfFirstColumn}. This proves that $P(\mathsf{M}\mathbf{x}|\mathbf{x})$ does not depend on $\vc{x}$.
\end{proof}

\subsection{Random Walk on Groups}
This section is mainly from \cite{Breuillard04}. A random walk on a group $(G,\cdot)$ is
\begin{IEEEeqnarray}{rCl}
S_K=X_K\cdot X_{k-1} \cdots X_1, \quad K = 1, 2, \ldots,
\end{IEEEeqnarray}
where $X_i\widesim{\text{i.i.d}}\mu(G)$.
If $X$ and $y$ are random variables on $G$ with PDF $\mu$ and $\nu$ respectively, then PDF of $X \cdot Y$ is $\mu \convolution \nu$, where $\convolution$ denotes the convolution. Hence, the PDF of $S_K$ is the $K$-th convolution power of $\mu$, denoted by $\mu^{*K}$.

The following theorem is proved by Kawada and It\^{o} for compact metric groups \cite{KawadaIto40}.  A more general version is proved by Stromberg  in \cite[Thm.~3.3.5]{Stromberg60} for Hausdorff groups, where the aperiodic condition is replaced with the normally aperiodic condition.

\begin{theorem}[Kawada-It\^{o} and Stromberg] \label{th:kawadaItoStromberg} 
Let $G$ be a compact Hausdorff groups and $H$ the  smallest closed subgroup of $G$ which contains $\mathcal{S}(\mu)$. Then,  $\lim \limits_{K \rightarrow \infty} \mu^{\convolution K} $ exists if  and only if  $\mu$  is a normally aperiodic probability measure on subgroup $H$. Moreover, if  this  limit exists, then it  is the Haar measure on $H$.
\end{theorem}

\begin{proof}
See \cite[Thm.~3.3.5]{Stromberg60}.
\end{proof}


\section{Continuity of Mutual Information}

\begin{lemma}
The mutual information of the fading channel \eqref{def:fadingChannel} and the SSFM channel defined in Section \ref{sec:ssfmDef} with $K$ segments is a continuous function of $K$ at $K\ra\infty$. 
\label{lemm:continuity}
\end{lemma}

\begin{proof}

Since $K$ is an integer, mutual information  $I(\vc X, \vc Y_K) $   is not a continuous function of $K$ when $K$ is finite. A small change in $\snr$ can change $K = \lfloor \sqrt[\delta]{\snr} \rfloor$ by one, and  the model by one segment.
However, 
mutual information is a continuous function at  $K\ra\infty$.

The proof is similar to the proof of the continuity of the output and conditional entropy in the zero-dispersion channel \cite[App.~I]{fahs2017capacity}, using the fact that the noise PDF, thus the output PDF induced by noise, vanishes exponentially. We sketch the steps for the SSFM channel.

The conditional PDF of one segment of SSFM is upper bounded
\begin{IEEEeqnarray*}{rCl}
p(\vc v_{2} |\vc x)\leq c'_1\sqrt{K} e^{-Kc_1\norm{\vc v_2- \vc x}^2},
\end{IEEEeqnarray*}
where $0<c_1,c'_1<\infty$ do not depend on $K$.
The conditional PDF of $2$ segments satisfies 
\begin{IEEEeqnarray*}{rCl}
p(\vc v_{3} |\vc x) &=& \int  p(\vc v_{3} |\vc v_2) p(\vc v_{2} |\vc x) d \vc v_2 
\\
&\leq& c'_2\sqrt{\frac{K}{2}} e^{-c_2\frac{K}{2}\norm{\vc v_3- \vc x}^2}.
\end{IEEEeqnarray*}
The conditional PDF of $K$ segments is upper bounded as
\begin{IEEEeqnarray}{rCl}
p(\vc y_{K} |\vc x)&\eqdef&p(\vc v_{K+1} |\vc x)
\nonumber\\&\leq& c'_k e^{-c_k\norm{\vc y_K - \vc x}^2}.
\label{eq:pdf-upper}
\end{IEEEeqnarray}
Alternatively, the exponential upper bound on the PDF \eqref{eq:pdf-upper} can be obtained from the PDF of the norm  which is known at the output \eqref{eq:conditionalFadingOutputNorm}.

We have
\begin{IEEEeqnarray*}{rCl}
\lim_{K\ra \infty} h\bigl(\vc Y_{K} | \vc X=\vc x \bigr) &=& -\lim_{K\ra\infty}  \int p(\vc y_K | \vc x)\log  p(\vc y_K | \vc x)\der \vc y_K
\nonumber\\
&\overset{(a)}{=}& -\int \lim_{K\ra\infty}   p(\vc y_K|\vc x)\log  p(\vc y_K|\vc x) \der \vc y_K
\\&=&h\bigl(\lim_{K\ra \infty} \vc Y_{K} | \vc X=\vc x \bigr),
\end{IEEEeqnarray*}
uniformly over input. Step $(a)$ is obtained from applying the dominated convergence theorem using \eqref{eq:pdf-upper}

The continuity of  $h\bigl(\vc Y_{K} \bigr)$ is shown similarly. 

\end{proof}


\section*{Acknowledgement}
This work has received funding from the European Research Council (ERC) under the European Union's 
Horizon 2020 research and innovation programme, Grant Agreement No. 805195.
The authors are greatly thankful to Emmanuel Breuillard for sharing his helpful ideas.

\bibliographystyle{IEEEtran}
\bibliography{refsOpticalFiber}

\end{document}